\documentclass[UTF8,13pt]{extarticle}
\usepackage{geometry}
\RequirePackage{amsthm,amsmath,amssymb}
\usepackage{algorithm}
\usepackage{algorithmic}
\usepackage{ulem}
\RequirePackage[colorlinks,citecolor=blue,linkcolor=blue,urlcolor=blue,pagebackref]{hyperref}
\usepackage{xr}
\usepackage{comment}
\usepackage{subcaption} 
\usepackage{graphicx} 
\usepackage{float}      
\usepackage{authblk}

\theoremstyle{plain}
\newtheorem{axiom}{Axiom}
\newtheorem{theorem}{Theorem}
\newtheorem{claim}{Claim}
\newtheorem{assumption}[claim]{Assumption}

\newtheorem{proposition}[axiom]{Proposition}

\theoremstyle{definition}

\newtheorem{example}{Example}
\newtheorem{remark}{Remark}

\def\argmin{\mathop \text{\rm{argmin}}}

\author[1]{Kun Ren}
\author[1]{Wen Su}
\author[2]{Li Liu}
\author[3]{Ian W. McKeague}
\author[4]{Xingqiu Zhao}

\affil[1]{Department of Biostatistics, City University of Hong Kong}
\affil[2]{School of Mathematics and Statistics, Wuhan University}
\affil[3]{Department of Biostatistics, Columbia University}
\affil[4]{Department of Applied Mathematics, The Hong Kong Polytechnic University}

\date{}
\begin{document}

\title{Flexible semiparametric modeling with application to Causal Inference}
\maketitle  

\begin{abstract}
This paper proposes a flexible new framework for constructing Neyman-orthogonal scores in semiparametric models involving infinite-dimensional nuisance parameters. 
While locally robust estimation is vital for integrating machine learning into econometrics, deriving orthogonal scores for complex models remains a major challenge. 
We provide explicit construction strategies for broad classes of settings.
The proposed framework ensures asymptotic normality of target parameter estimators in a way that does not depend on the method used to construct the nuisance parameter estimators, provided they are $o_p(n^{-1/4})$-consistent. 
We apply the proposed methodology to causal inference with a binary instrumental variable, developing a novel, robust estimator for treatment effects. 
Numerical studies demonstrate that our approach significantly outperforms naive alternatives in finite samples. 
An empirical application to the Oregon Health Insurance Experiment illustrates the framework’s utility in providing robust causal evidence.
\end{abstract}

\section{Introduction}\label{intro}
In semiparametric modeling, Neyman-orthogonal scores allow the construction of locally robust estimators of finite-dimensional target parameters, even in the presence of complex infinite-dimensional nuisance components. Such nuisance components can be approximated using machine learning methods, but typically a bias correction step is needed for the target parameter estimator (\cite{10.1111/ectj.12097}).  This can sometime be done by constructing  scores of the target parameters that are Neyman-orthogonal, in the sense that they have zero derivative with respect to the nuisance components, combined with a cross-fitting technique.
Asymptotic normality of target estimators constructed from empirical orthogonal scores can then be established provided the nuisance estimators converge at a rate of $o_p(n^{-1/4})$, see
\cite{10.1111/ectj.12097}, \cite{chernozhukov2022locally}, \cite{chernozhukov2022debiased}, \cite{chernozhukov2022automatic}, and \cite{chernozhukov2023simple}.

The present paper studies the question of how  to systematically derive Neyman-orthogonal scores.
Although many methods for constructing such scores have been developed for diverse situations, including direct computation when nuisance parameters are finite-dimensional (\cite{10.1111/ectj.12097}), application of the Riesz representation theorem (\cite{chernozhukov2022locally},\cite{chernozhukov2022debiased},\cite{chernozhukov2022automatic}), adjustment of original scores using first-step influence functions (FSIFs) (\cite{10.1111/ectj.12097}, \cite{chernozhukov2022locally}), and others, these developments are limited to special statistical models where target parameters are linear functionals of regression functions, or focus on well-investigated examples. 
Providing explanations from the perspective of Neyman orthogonality, as well as deriving the orthogonal scores, however, is often a significant hurdle in applications (\cite{farrell2020deep}). 
These limitations substantially hinder the widespread adoption of debiased machine learning.

Our main contributions are twofold, as explained in the following paragraphs.
First, we provide systematic and easy-to-implement strategies for constructing Neyman-orthogonal scores from the original scores, across a broad class of semiparametric models that may involve multi-stage estimation procedures and multiple infinite-dimensional nuisance parameters.
Then for estimators constructed by minimizing empirical risks,
we derive influence functions for the target parameters via asymptotic normality theory for functional $M$-estimators.  Throughout we make minimal assumptions on model structure, thereby naturally demonstrating that the resulting influence functions can serve as Neyman-orthogonal scores.

To specify our orthogonal-score construction strategy, the relationship between target and nuisance parameters will be studied 
separately for the following two cases:
\begin{itemize}
\item[]\textbf{Case I}. All unknown parameters are characterized by simultaneously minimizing a specific criterion function (subsection~\ref{S1});
\item[]\textbf{Case II}. Target and nuisance parameters require sequential identification (subsections~\ref{S2} and \ref{SI1}).
\end{itemize}

In  Case I,  functional central limit theory for working $M$-estimators is used to reveal how nuisance estimators influence the limiting distribution of estimators of the target parameters. This influence is determined by the Gateaux derivative of the criterion function with respect to the nuisance function in a given direction. 
We then demonstrate that selecting specific directions and constructing corresponding influence functions that eliminate the first-order impact of nuisance estimators is equivalent to deriving Neyman-orthogonal scores.
In Case II with a single nuisance function, we develop a flexible approach to derive closed form expressions for FSIFs, facilitating the construction of robust estimators based on Neyman orthogonality (subsection~\ref{S2}).
Importantly, we extend construction to settings involving multiple nuisance functions (subsection~\ref{SI1}), yielding results that are novel compared to the existing literature.
Moreover, we clarify the role of additional infinite-dimensional nuisance parameters in the orthogonal-score construction, provide their closed-form expressions, and establish a theoretical guarantee that no efficiency loss occurs. 

Our second contribution is the application of the proposed methodology to instrumental variable (IV) models within a potential outcomes causal inference framework.
We propose a robust estimation and inference method for the causal effect among the ``compliers'', as defined by \cite{angrist1996identification}. 
The proposed estimator is motivated by the identification results of \cite{abadie2003semiparametric}, which rely on instrumental propensity score weighting. 
By establishing Neyman orthogonality to augment these weighting methods, we construct a novel score function for the causal parameter of interest.
We demonstrate that the cross-fitted estimator of the causal effect exhibits asymptotic normality and robustness to misspecification of nuisance functions.
Simulation studies demonstrate that the finite-sample performance of our proposed estimator outperforms existing alternatives.
To the best of our knowledge, the integration of Neyman-orthogonality with nonparametric multi-stage estimation is novel in the causal inference literature, providing a flexible and generalizable framework for locally robust estimation and inference in IV settings.

The remainder of the article is organized as follows.
Section~\ref{SI} presents insightful and systematic methods for constructing Neyman-orthogonal scores in various scenarios, illustrated with concrete examples. 
Section~\ref{Sec4} applies the framework to an observational study with a valid binary instrument, developing a robust approach to estimate and infer population-level treatment effects.
Section~\ref{simulation} shows comprehensive empirical results in simulation studies.
Section~\ref{Sec5} presents an empirical application to the Oregon Health Insurance Experiment (OHIE) dataset.
Section~\ref{Sec6} concludes with a summary of findings and implications for future research.
The Appendix provides all technical proofs of results in Sections \ref{SI} and \ref{Sec4}.

\section{Locally robust semiparametric estimation and inference}\label{SI}
This section develops a novel, model-free approach for constructing Neyman-orthogonal scores.
While there are several influential studies along these lines, such as \cite{10.1111/ectj.12097, chernozhukov2022locally}, our primary aim here  is to demonstrate how asymptotic theory for functional $M$-estimators provides a unified approach that simplifies and systematizes the construction of such scores.

We begin with the setup.
Let the observed data $S=\{W_i\}_{i=1}^{n}$ consist of i.i.d.\ copies of the random vector $W$ with probability law $P$.
Define the population and empirical measures for any measurable function $g$ as $\mathbb Pg(W)=\int g(w)dP(w)$ and $\mathbb{P}_ng(W)=n^{-1}\sum_{i=1}^{n}g(W_{i})$.

Let $X \in \mathcal{X}$ be a sub-vector of the random vector $W$.
The parameter of interest is a finite-dimensional vector $\beta_0$, which is assumed to be a scalar for convenience.
It can be identified by a score condition involving an infinite-dimensional nuisance parameter $f_0 \in \mathcal{F}$:
\begin{equation} \label{naive_moment}
\mathbb{P} \psi(\beta_0, f_0; W)  = 0,
\end{equation}
where $\psi$ is a known score function, and the nuisance parameter space $\mathcal{F}$ is a normed space consisting of real-valued functions mapping from $\mathcal{X}$ to $\mathbb{R}$.

Throughout this paper, we need a regularity condition that permits the interchange of differentiation and expectation.
That is, to derive the Gateaux derivative of a real-valued functional $(\beta, f)\mapsto \mathbb{P} \psi(\beta, f; W)$ (\cite{van2000asymptotic}), we assume that the function $\psi$ is continuously differentiable with respect to its arguments.
The first partial derivative of $\mathbb{P} \psi(\beta, f; W)$ with respect to $f$ at $(\beta_0,f_0)$ in a direction $h \in  \mathcal{F}$ is 
\begin{align*}
& \partial_f \mathbb{P}  \psi(\beta_0,f_0;W)[h] 
:=  \lim_{\delta\rightarrow 0} \mathbb{P} \frac{\psi(\beta_0,f_0 + \delta h;W) - \psi(\beta_0,f_0;W)}{\delta} \\
= &\mathbb{P}  \lim_{\delta\rightarrow 0} \frac{\psi(\beta_0,f_0 + \delta h;W) - \psi(\beta_0,f_0;W)}{\delta}=  \mathbb{P}   \{\partial_f \psi(\beta_0,f_0;W)\cdot h\},
\end{align*}
where the second equality follows from our regularity assumption, and $\partial_{f} \psi(\beta_0, f_0; w)$ denotes the first-order partial derivative of $\psi$ with respect to its second argument at $(\beta_0, f_0)$.
Furthermore, we use $\partial^2_{\beta f} \psi(\beta_0, f_0; w)$ to denote the second-order mixed partial derivative of the function $\psi$ with respect to its first and second arguments at $(\beta_0, f_0)$, and other partial derivatives are defined analogously.

We provide the following example to illustrate the above notation. 
Suppose that the observed data is $W = (X, D,Y)$ consisting of a $p$-dimensional covariate vector $X$, a treatment indicator $D$ and the response variable $Y$.
Under the potential outcome framework (\cite{rubin1974estimating}), the target parameter is the mean potential outcome under treatment $\beta_0:=E Y(1) $, and the nuisance function is $f_0 := \log[g_0 / (1 - g_0)]$, where $g_0(X) =P(D=1 \mid X )$ denotes the propensity score.
The inverse probability weighting method proposed  by \cite{rosenbaum1983central} is based on the identity:
$\beta_0 - \mathbb{P}  [DY (1+e^{-f_0(X)}) ] = 0.$
In this case, for a realization $w=(x,d,y)$ of the random vector $W$, the score function of $\beta_0$ is $\psi(\beta, f; w) = \beta - yd(1+e^{-f(x)})$. 
Then, its partial derivatives are $\partial_{\beta} \psi(\beta, f; w) = 1$, $\partial_f \psi(\beta, f; w) = yd e^{-f(x)}$, $\partial_{f \beta}^2 \psi(\beta, f; w) =0$ and $\partial_{ff}^2 \psi(\beta, f; w) = -yde^{-f(x)}$.

A score function for $\beta_0$, $\psi^*(\beta,f;w)$ is Neyman-orthogonal with respect to the nuisance function $f_0$ if it satisfies $\mathbb{P}  \psi^*(\beta_0,f_0;W)=0$, $\mathbb{P}  |\psi^*(\beta_0,f_0;W)|^2 < \infty$ and for any direction $f\in\mathcal{F}$, it obeys the Neyman-orthogonal condition at $(\beta_0,f_0)$:
$$\partial_{f} \mathbb{P} \psi^*(\beta_0,f_0;W)[f-f_0] = 0.\nonumber$$
This orthogonality condition implies that the score function $\psi^*$ is locally insensitive to perturbations of the nuisance function $f$ around its true value $f_0$.
Therefore, the score function $\psi^*$ is considered locally robust.
Based on the original score condition \eqref{naive_moment}, we propose several systematic methods for constructing such orthogonal scores that apply in various semiparametric scenarios.

\subsection{Semiparametric model with a single nuisance function}\label{SI2}
\subsubsection{Fully coupled semiparametric model} \label{S1}
Consider a setting, in which the parameter of interest, $\beta_0 \in \mathbb{R}$, and an infinite-dimensional nuisance function, $f_0 \in  \mathcal{F}$, are jointly identified as the solution to an $M$-estimation problem.
Specifically, the true parameters $(\beta_0, f_0)$ satisfy
\begin{equation}\label{M}
(\beta_0, f_0) = \underset{(\beta, f) \in \mathbb{R} \times  \mathcal{F}}{\argmin} \mathbb{P} m( \beta, f;W),
\end{equation}
where $m$ is a known criterion function.
This formulation implies that the identification of $\beta_0$ is intrinsically linked to the correct specification of $f_0$.
Such joint $M$-estimation problems are common in econometrics (\cite{newey1999nonparametric}), panel count data analysis (\cite{10.1214/009053607000000181}) and survival analysis (\cite{liu2022nonparametric, zhong2022deep}). 
We now derive the Neyman-orthogonal score for $\beta_0$.

First, we consider a working $M$-estimator, defined as
\begin{equation*}
(\hat{\beta}_n^M, \hat{f}_n^M) = \argmin_{(\beta, f) \in \mathbb{R} \times  \mathcal{F}_n} \mathbb{P}_n m(\beta, f; W),
\end{equation*}
where $ \mathcal{F}_n \subseteq  \mathcal{F}$ is a sieve space used to approximate the nuisance function.
The key insight is that the influence function of $\hat{\beta}_n^M$ can serve as a score function of $\beta_0$, which obeys the Neyman-orthogonal condition.  

To derive this influence function, we provisionally assume that the estimator $(\hat{\beta}_n^M, \hat{f}_n^M)$ exhibits asymptotic functional normality.
The asymptotic normality of this working estimator serves solely as a heuristic tool to guide the construction of the orthogonal score for $\beta_0$. 
The orthogonality of the influence function does not depend on the assumptions imposed on $(\hat{\beta}_n^M, \hat{f}_n^M)$, and the existence of such a working estimator is not required.
We make the following assumptions to derive the asymptotic expansion of $\hat{\beta}_n^M$. 

\begin{assumption}\label{as_pre}
For any direction $h \in \mathcal{F}$, we impose the following conditions:
\begin{itemize}
\item[(a)] (Stochastic equicontinuity)
\begin{align*}
&\sqrt{n}(\mathbb{P}_n-\mathbb{P})\{ \partial_{\beta} m(\hat{\beta}_n^M, \hat{f}_n^M; W) + \partial_{f} m(\hat{\beta}_n^M, \hat{f}_n^M; W)h \} \\
&- \sqrt{n}(\mathbb{P}_n-\mathbb{P}) \left\{ \partial_{\beta} m(\beta_0, f_0; W) + \partial_{f} m(\beta_0, f_0; W)h \right\} = o_p(1).
\end{align*} 

\item[(b)] {(Score condition)}
\begin{enumerate}
\item[(i)] $\mathbb{P} \{ \partial_{\beta} m(\beta_0, f_0; W) + \partial_{f} m(\beta_0, f_0; W) h  \} = 0$.
\item[(ii)] $\mathbb{P}_n \{ \partial_{\beta} m(\hat{\beta}_n^M, \hat{f}_n^M; W) + \partial_{f} m(\hat{\beta}_n^M, \hat{f}_n^M; W) h  \} = o_p(n^{-1/2})$.
\end{enumerate}

\item[(c)] (Smoothness)
The map $(\beta, f) \mapsto \mathbb{P} \{ \partial_{\beta} m(\beta, f; W) + \partial_{f} m(\beta, f; W)h \}$ is Fréchet-differentiable at $(\beta_0, f_0)$.

\item[(d)] (Second-order condition)
\begin{align*}
&\mathbb{P} \{ \partial_{\beta} m(\hat{\beta}_n^M, \hat{f}_n^M; W) + \partial_{f} m(\hat{\beta}_n^M, \hat{f}_n^M; W) h  - \partial_{\beta} m(\beta_0, f_0; W) - \partial_{f} m(\beta_0, f_0; W) h  \} \\
& - \mathbb{P}[ \{ \partial^2_{\beta \beta} m(\beta_0, f_0; W) + \partial^2_{f \beta} m(\beta_0, f_0; W) h  \} (\hat{\beta}_n^M - \beta_0) ] \\
& - \mathbb{P} \{ \partial^2_{\beta f} m(\beta_0, f_0; W) (\hat{f}^M_n-f_0) + \partial^2_{f f} m(\beta_0, f_0; W) h (\hat{f}^M_n-f_0) \}   = o_p(n^{-1/2}).
\end{align*}

\end{itemize}
\end{assumption}

\begin{remark}[Assumption \ref{as_pre}]
Assumption \ref{as_pre}(a) amounts to asymptotic stochastic equicontinuity of the empirical process term. 
This condition depends on the entropy of the function classes containing $(\hat{\beta}_n^M,\hat{f}_n^M)$ and $(\beta_0,f_0)$ (\cite{vaart1996weak, https://doi.org/10.1111/1468-0262.00461}).

Assumption \ref{as_pre}b(i) {represents a score condition derived from the $M$-estimation problem \eqref{M}. }
Assumption \ref{as_pre}b(ii) is satisfied if the empirical score is sufficiently close to zero.

Assumption \ref{as_pre}(c) requires sufficient smoothness of the objective function.

Assumption \ref{as_pre}(d) dominates the remainder term from the Taylor expansion, combining with (c)
The left-hand side of the equation in (d) has the order of $|\hat{\beta}_n^M-\beta_0|^2 + \mathbb{P}|\hat{f}_n^M(X) - f_0(X)|^2$. 
Therefore, (d) holds if $\hat{\beta}_n^M$ and $\hat{f}_n^M$ converge at rates faster than $o_p(n^{-1/4})$.

\end{remark}

\begin{proposition}\label{prop0}
Under Assumption \ref{as_pre}, the functional of the $M$-estimator, $(\hat{\beta}_n^M, \hat{f}_n^M)$ with any direction $h\in\mathcal{F}$ has the asymptotic expansion:
\begin{align} \label{co}
&-\sqrt{n} \int\{\partial^2_{\beta \beta} m(\beta_0, f_0;w) + \partial^2_{f \beta} m(\beta_0, f_0;w)h \} (\hat{\beta}_n^M - \beta_0) dP(w) \nonumber\\
& -\sqrt{n} \int [\partial^2_{\beta f} m(\beta_0, f_0;w) \{\hat{f}^M_n(x)-f_0(x)\}  \nonumber\\
& \hspace{17mm}+ \partial^2_{f f} m(\beta_0, f_0;w)h(x) \{\hat{f}^M_n(x)-f_0(x)\} ]  dP(w) \nonumber\\
=&\sqrt{n}(\mathbb{P}_n-\mathbb{P})\{\partial_{\beta} m(\beta_0, f_0;W) + \partial_{f} m(\beta_0, f_0;W)h(X) \} + o_p(1).
\end{align}
\end{proposition}

Equation \eqref{co} implies that the limiting distribution of $\sqrt{n} (\hat{\beta}_n^M - \beta_0)$ depends on $\hat{f}_n^M$ solely through the term
$\sqrt{n} \int \{ \partial^2_{\beta f} m(\beta_0, f_0; w)(\hat{f}^M_n-f_0)  + \partial^2_{f f} m(\beta_0, f_0; w)h(\hat{f}^M_n-f_0)   \} dP(w)$.
If there exists a direction $h_0 \in  \mathcal{F} $ such that this term equals zero, then $\hat{f}_n^M$ has no first-order impact on $\sqrt{n} (\hat{\beta}_n^M - \beta_0)$ along any path. Therefore, the resulting influence function is orthogonal with respect to the nuisance function. Thus, we can demonstrate that a valid direction $h_0$ that satisfies  the considered condition is
\begin{equation}\label{h^*}
h_0(x) = -\frac{E[\partial^2_{\beta f} m(\beta_0, f_0;W)\mid X=x]}{E[\partial^2_{f f} m(\beta_0, f_0;W)\mid X=x]}, 
\end{equation}
where the denominator is assumed to be bounded away from zero so that $h_0$ is well-defined.
With such a designed $h_0$, the $M$-estimator $\hat{\beta}_n^M$ achieves asymptotic normality, satisfying
\begin{align*}
&-\sqrt{n} \int\{\partial^2_{\beta \beta} m(\beta_0, f_0;w) + \partial^2_{f \beta} m(\beta_0, f_0;w) h_0  \}  dP(w) (\hat{\beta}_n^M- \beta_0)\nonumber\\
= &\sqrt{n}(\mathbb{P}_n-\mathbb{P}) \{\partial_{\beta} m(\beta_0, f_0;W) + \partial_{f} m(\beta_0, f_0;W) h_0(X) \} + o_p(1).
\end{align*}

It is worth noting that when $m$ is a negative log-likelihood function, the $M$-estimator $\hat{\beta}_n^M$ achieves the asymptotic information bound (\cite{van1991differentiable, bickel1993efficient}), and the centered influence function $\partial_{\beta} m(\beta, f;w) + \partial_{f} m(\beta, f;w)h$ is the so-called efficient score function with nuisance functions $f_0$ and $h_0$.

Before continuing the discussion, we mention Algorithm \ref{alg:cross_fitting} that specifies a cross-fitting estimation procedure for the target parameter $\beta_0$ via a user-specified score function $\psi$. 
This procedure helps relax the Donsker conditions on $\mathcal{F}_n$ and $\mathcal{F}$ via the use of Neyman-orthogonal score functions; a more detailed discussion of this will be provided later in Remark \ref{RE1}. 
We assume that the nuisance parameter space $\mathcal{F}$ is a set of functions on $\mathcal{X}$ that are uniformly bounded by a positive constant $B$. 
Furthermore, we claim that an estimator of $  (f_0,h_0)$, $ (\hat{f}, \hat{h})$ is $o_p(n^{-1/4})$-consistent if both $\hat{f}$ and $\hat{h}$ converge at a rate of $o_p(n^{-1/4})$.

\begin{theorem}[Neyman orthogonality in the fully coupled semiparametric model] \label{NO_score}
Suppose the target parameter $\beta_0$ and the nuisance parameter $f_0$ are jointly characterized by \eqref{M}. 
Assume the following conditions hold:
\begin{itemize}
\item[(a)] The function $m$ is thrice continuously differentiable with respect to its arguments.
\item[(b)] The function $h_0$ in \eqref{h^*} is well-defined.
\item[(c)] $\partial_{\beta}\mathbb{P} m(\beta_0, f_0; W) = 0$, and $\partial_{f} \mathbb{P}m(\beta_0, f_0; W)[h] = 0$ for any direction $h \in \mathcal{F}$.
\end{itemize}
Then, the following properties hold:
\begin{itemize}
\item[(i)] The function 
\begin{align}\label{def_NO}
\psi^*(\beta, f, h; w) = \partial_{\beta} m(\beta, f; w) + \partial_{f} m(\beta, f; w)h  
\end{align}
serves as a valid score function for $\beta_0$ with nuisance functions $\theta_0 = (f_0, h_0)$. 
Furthermore, it satisfies the Neyman-orthogonal condition at $(\beta_0, \theta_0)$ with respect to $\theta_0$.
\item[(ii)] Assume that estimators $(\hat{\theta}_{n_1},\hat{\theta}_{n_2},\hat{\beta}_{n_1},\hat{\beta}_{n_2})$ constructed in Algorithm \ref{alg:cross_fitting} based on $\psi^*$ in \eqref{def_NO} are $o_p(n^{-1/4})$-consistent. Then, the output of Algorithm \ref{alg:cross_fitting}, denoted by $\hat{\beta}_n^{R}$, has the asymptotic expansion
$$-\sqrt{n}(\hat{\beta}_n^{R} - \beta_0)\partial_{\beta}\mathbb{P}\psi^*(\beta_0 ,f_0,h_0;W) = \sqrt{n}(\mathbb{P}_n-\mathbb{P})\psi^*(\beta_0, f_0,h_0;W) + o_p(1).$$
\end{itemize}
\end{theorem}

The proof of Theorem \ref{NO_score} is given in the Appendix.
\begin{algorithm}[h]
\caption{Cross-fitting Estimation}
\label{alg:cross_fitting}
\begin{algorithmic}[1]
\REQUIRE Pooled sample $ S= \{W_i\}_{i=1}^n$, score function $\psi$ with nuisance function $\theta_0$.
\ENSURE The target estimator $\hat{\beta}_n$.

\STATE \textbf{Step I: Data Splitting} \\
Randomly split the  pooled sample $S$ into two subsets $S_1 = \{W_i\}_{i=1}^{n_1}$ and $S_2 = \{W_i\}_{i=1}^{n_2}$ with sample sizes $n_1$ and $n_2$ such that $n_1 = n_2$.

\FOR{$(j, k) \in \{(1, 2), (2, 1)\}$}
\STATE \textbf{Step II: Nuisance Function Estimation} \\
Construct an estimator $\hat{\theta}_{n_j} $ of nuisance function $\theta_0$ using $S_j$ only.

\STATE \textbf{Step III: Target Parameter Estimation} \\
Solve the following estimating equation using the data in subset $S_k$ to obtain $\hat{\beta}_{n_k}$:
$  \sum_{i=1}^{n_k} \psi(\hat{\beta}_{n_k}, \hat{\theta}_{n_j} ; W_i) = 0.$
\ENDFOR

\STATE \textbf{Step IV: Cross Fitting} \\
Compute the target estimator by averaging the results from both folds:
$$\hat{\beta}_n = \frac{1}{2} \left( \hat{\beta}_{n_1} + \hat{\beta}_{n_2} \right).$$

\RETURN $\hat{\beta}_n$
\end{algorithmic}
\end{algorithm} 
\begin{remark}[The role of Assumption \ref{as_pre} and the functional asymptotic expansion in \eqref{co} in establishing Neyman-orthogonal score functions]
Assumption \ref{as_pre} is proposed to derive the asymptotic expansion of $\sqrt{n}(\hat{\beta}_n^M-\beta_0)$ in \eqref{co}. 
Based on this expansion, we choose a special weight function $h_0$, construct a reasonable influence function of $\hat{\beta}_n^M$, $\psi^*$ that is insensitive to the first-order effect of $\sqrt{n}(\hat{f}_n^M-f_0)$, and prove that it can serve as a valid score function for $\beta_0$ in Theorem \ref{NO_score}. 
It is important noting that the Neyman orthogonality of $\psi^*$ holds independently of Assumption \ref{as_pre}.

Comparing Assumption \ref{as_pre} with the conditions outlined in Theorem \ref{NO_score}(ii), we conclude that the asymptotic normality of the estimator $\hat{\beta}_n^{R}$ is established under milder conditions.
\end{remark}

\begin{remark}[Robust inference via orthogonality and cross-fitting]\label{RE1}
Deriving the limiting distribution of the estimator of $\beta_0$ typically requires that the nuisance estimators belong to Donsker classes (\cite{10.1111/ectj.12097}).
However, combining with the Neyman-orthogonal conditions, we can relax this restrictive condition by employing the cross-fitting strategy outlined in Algorithm \ref{alg:cross_fitting}. 
Further details are provided in the proof of Theorem \ref{NO_score}.

On the other hand, combining standard semiparametric inference with machine learning methods faces a challenge: Assumption \ref{as_pre}(b)(ii) is usually not satisfied. 
As discussed by \cite{yan2025semiparametric}, when the nuisance $M$-estimator $\hat{f}_{n}^M$ is constructed via deep learning, the approximation capability of the neural network tangent space should be carefully analyzed to ensure this assumption holds. 
On the contrary, as shown in the proof of Theorem \ref{NO_score}, the proposed approach decouples the estimations of $f_0$ and $\beta_0$, thereby eliminating the reliance on Assumption \ref{as_pre}(b)(ii) when establishing the asymptotic normality of $\hat{\beta}_n^R$.
\end{remark}

\begin{remark}[Comparison with \cite{10.1111/ectj.12097}]\label{RE2}
For \eqref{M}, \cite{10.1111/ectj.12097} provided a concentrating-out approach to construct Neyman-orthogonal score functions.
Letting
\begin{align*}
f_{\beta} = \argmin_{f\in \mathcal{F}} \mathbb{P} m(\beta,f;W),
\end{align*}
\cite{10.1111/ectj.12097} obtained the Neyman score function $\psi^*$:
\begin{align*}
\psi^* (\beta, f_{\beta};W) = - \frac{d m(\beta,f_{\beta};W)}{d \beta}.
\end{align*}

However, the issue is that deriving a valid $f_{\beta}$ is usually not easy, particularly when $f_0$ is not a regression function, a scenario in which $f_{\beta}$ typically lacks a closed-form expression.
For instance, in partial linear Cox model (\cite{zhong2022deep}), where the conditional hazard function given covariates $(X,Z) \in \mathbb{R}\times\mathbb{R}^p$ is expressed as
$\lambda(t\mid X ,Z ) =\lambda_0(t) \exp\{\beta_0 X + f_0(Z)\} $
with univariate function $\lambda_0(t)$, and the criterion in \eqref{M} then becomes the negative log partial likelihood.
Therefore, calculating $f_{\beta}$ is extremely challenging, so that the concentrating-out approach is no longer applicable.
Instead, the proposed model-free method for developing $\psi^*$ neither suffers from such a problem nor assumes that nuisance parameters are regression functions.
Theorem \ref{NO_score} can be naturally extended to cases involving multiple nuisance functions, thus providing a more flexible and feasible approach for orthogonality construction.
\end{remark}

Theorem \ref{NO_score} establishes a Neyman-orthogonal score condition for $\beta_0$.
The proposed estimation offers two primary advantages:
\begin{enumerate}
\item[(i)] Although we introduce an additional nuisance function $h_0$, the proposed score function $\psi^*$ significantly relaxes the conditions required for achieving asymptotic normality of $\hat{\beta}_n^R$ through (a) cross-fitting technology; and (b) constructing $o_p(n^{-1/4})$-consistent estimators for all nuisance functions.

\item[(ii)] The estimators $\hat{\beta}_n^R$ and $\hat{\beta}_n^M$ have the same asymptotic variance.
Thus, when the working estimator $\hat{\beta}_n^M$ is semiparametrically efficient, the robust estimator $\hat{\beta}_n^R$ also attains the information bound.
Here, the augmented scores become efficient Neyman-orthogonal scores.
\end{enumerate}

To illustrate the application of Theorem \ref{NO_score}, we revisit the example considered by \cite{10.1111/ectj.12097} and apply the proposed Neyman-orthogonal method to the partially linear model.

\begin{example}[Partially linear regression (\cite{10.1111/ectj.12097})]\label{ex3}
Let $Y\in\mathbb{R}$ be response variable along with covariates $D\in\mathbb{R}$ and $X\in\mathbb{R}^p$.
Under partially linear model,
\begin{equation}\label{PL}
Y=\beta_0 D + f_0(X) + \epsilon,\ E[\epsilon\mid D,X]=0,
\end{equation}
we consider two estimation approaches to $\beta_0$ and provide the limiting distributions of associated estimators.
\begin{enumerate}
\item[(i)] $M$-estimation: 
$(\hat{\beta}_n^M,\hat{f}_n^M) := \argmin_{(\beta,f) \in\mathbb{R}\times \mathcal{F}_n} \mathbb{P}_n\{\beta D+ f(X) - Y\}^2$. 
By Proposition \ref{prop0}, under some basic assumptions, we have that for any $h\in \mathcal{F}$,
\begin{align*}
&-\sqrt{n} \mathbb{P}\{ D^2 + D h(X)  \} (\hat{\beta}_n^M - \beta_0) - \sqrt{n} \mathbb{P} [\{ D + h(X) \} \{\hat{f}_n^M(X)-f_0(X)\}] \nonumber\\
= &\sqrt{n}(\mathbb{P}_n-\mathbb{P})[\{ D+h(X) \} \{\beta_0D + f_0(X) -Y\}] + o_p(1).
\end{align*}
Therefore, choosing the direction $h_0 (x)= -E( D\mid X=x )$, we can conclude that
\begin{align*}
&-\sqrt{n}(\hat{\beta}_n^M -\beta_0) \mathbb{P}\{ D^2 - D E( D\mid X ) \} \nonumber\\
= & \sqrt{n}(\mathbb{P}_n-\mathbb{P}) [\{ D-E( D\mid X )\}  \{ \beta_0 D + f_0(X) -Y \}] +o_p(1).
\end{align*}
Then, the asymptotic normality of $\hat{\beta}_n^M$ holds when $\mathbb{P}\{ D^2 - D E( D\mid X ) \} \neq 0$.

\item[(ii)] Locally robust estimation:
Theorem \ref{NO_score} (i) provides a Neyman-orthogonal score condition with respect to $\beta_0$:
$$\mathbb{P}\psi^*(\beta_0,h_0,f_0;W) = \mathbb{P}[\{D-E(D\mid X)\} \{\beta_0D + f_0(X) -Y \} ] = 0 $$
with the nuisance function $h_0(x) = -E( D\mid X=x )$.
Under conditions in the Theorem \ref{NO_score} (ii), the estimator $\hat{\beta}_n^R$ obtained via Algorithm \ref{alg:cross_fitting} is asymptotically normal 
\begin{align*}
&-\sqrt{n}(\hat{\beta}_n^R-\beta_0) \mathbb{P}\{ D^2 - D E (D\mid X ) \} \nonumber\\
= &\sqrt{n}(\mathbb{P}_n-\mathbb{P}) [\{D-E(D\mid X )\}  \{\beta_0 D + f_0(X) -Y\} ] +o_p(1).
\end{align*}
\end{enumerate}

Although such an example was comprehensively investigated by \cite{10.1111/ectj.12097}, following the analysis above, we can gain a novel insight.

When $\epsilon$ is independent of $(D,X)$ and follows a normal distribution, we can demonstrate that the estimators $\beta_n^M$ and $\beta_n^R$ are semiparametrically efficient under regularity conditions. However, when the distribution of $\epsilon$ is unknown, achieving an efficient estimation of $\beta_0$ requires estimating the probability density function of $\epsilon$. This maximum likelihood estimation approach is more complex than the standard least squares criterion.
But in this situation, Theorem \ref{NO_score} is still useful as target and nuisance parameters can be identified through a likelihood function simultaneously.
It establishes that the efficient score for $\beta_0$, denoted by $\psi^*(\beta,h,\lambda;w):=(d+h)\lambda'(y-\beta d - f(x))$, satisfies the Neyman-orthogonal condition, yielding:
\begin{align*}
\mathbb{P}\psi^*(\beta_0,f_0,\lambda_0;W) = \mathbb{P} \left[ (D+h_0 (X)) \lambda_0'(Y-\beta_0 D- f_0(X)) \right] = 0,
\end{align*}
where $\lambda_0$ denotes the negative log-density function of $\epsilon$, and $\lambda'$ and $\lambda_0'$ denote the first derivatives of $\lambda$ and $\lambda_0$, respectively.


\end{example}

\subsubsection{Partially decoupled semiparametric model}\label{S2}
Again, we consider a semiparametric model where the parameter of interest,  $\beta_0\in\mathbb{R}$, satisfies the score condition 
\begin{align}\label{ori_score1}
\mathbb{P}  \psi(\beta_0, f_0; W)  = 0.
\end{align} 
However, the nuisance function $f_0$ cannot be identified through $\psi$ any more.
Instead, it can be characterized as a minimizer under a criterion with the known function $m_1$ such that
\begin{equation}\label{ori_M}
f_0 = \argmin_{f \in  \mathcal{F}} \mathbb{P} m_1(f; W).
\end{equation}
This partially decoupled structure is prevalent in causal inference for observational studies, where potential outcomes are identified through covariate adjustment methods, such as inverse probability weighting.
In particular, $f_0$ may represent the propensity score that is characterized by minimizing cross-entropy loss; $\beta_0$ denotes the causal effect of interest that involves the potential outcomes, which can be identified using $f_0$ under standard assumptions (\cite{bang2005doubly, firpo2007efficient}). 

Let  
\begin{equation}\label{weight}
h_0(x) = -\frac{E[\partial_{f}  \psi(\beta_0, f_0;W)\mid X=x]}{E[\partial^2_{ff} m_{1}( f_0;W)\mid X=x]}.  
\end{equation} 
\begin{theorem}[Neyman orthogonality in the partially decoupled semiparametric model] \label{NO2_score}
Suppose the target parameter $\beta_0$ is identified via the score condition \eqref{ori_score1}, while the nuisance function $f_0$ is independently characterized by \eqref{ori_M}.
Assume the following conditions hold:
\begin{itemize}
\item[(a)] The functions $\psi$ and $m_1$ are twice and thrice continuously differentiable with respect to their arguments, respectively.
\item[(b)] The function $h_0$ in \eqref{weight} is well-defined.
\item[(c)] $\partial_{f}\mathbb{P} m_1(f_0; W)[h] = 0$ for any direction $h \in \mathcal{F}$.
\end{itemize}
Then, the following properties hold:
\begin{itemize}
\item[(i)] The function 
\begin{align}\label{def_NO2}
\psi^*(\beta, f, h; w) = \psi(\beta,f;w) + \partial_{f} m_{1}( f;w) h 
\end{align}
serves as a valid score function for $\beta_0$ with nuisance parameters $\theta_0 = (f_0, h_0)$, and it satisfies the Neyman-orthogonal condition at $(\beta_0, \theta_0)$ with respect to $\theta_0$.
\item[(ii)] Assume that estimators $(\hat{\theta}_{n_1},\hat{\theta}_{n_2},\hat{\beta}_{n_1},\hat{\beta}_{n_2})$ constructed in Algorithm \ref{alg:cross_fitting} based on $\psi^*$ in \eqref{def_NO2} are $o_p(n^{-1/4})$-consistent.
Then, the output of Algorithm \ref{alg:cross_fitting}, denoted by $\hat{\beta}_n^{R}$, has the asymptotic expansion
$$-\sqrt{n}(\hat{\beta}_n^{R} - \beta_0)\partial_{\beta}\mathbb{P}\psi^*(\beta_0 ,f_0,h_0;W) = \sqrt{n}(\mathbb{P}_n-\mathbb{P})\psi^*(\beta_0, f_0,h_0;W) + o_p(1).$$
\end{itemize}
\end{theorem}
As proof of Theorem \ref{NO2_score} is similar to that of Theorem \ref{NO_score}, it is omitted. 
\begin{remark}[Comparison with   \cite{chernozhukov2022locally}]\label{RE3} 
\cite{chernozhukov2022locally} demonstrated that the existence of first stage influence function $\phi(\beta , \theta ; w)$ with nuisance function $\theta_0$,  is essential for constructing orthogonal scores. The function $\phi$ satisfies that $\mathbb{P}\phi(\beta_0,  \theta_0; W) =0$, $\mathbb{P}|\phi(\beta_0, \theta_0; W)|^2<\infty$, and
\begin{equation}
\sqrt{n}\mathbb{P}\{ \psi(\beta_0, \hat{f}_n; W) - \psi(\beta_0, f_0; W) \} =  \sqrt{n}(\mathbb{P}_n-\mathbb{P}) \phi(\beta_0, \theta_0; W ) + o_p(1). \nonumber
\end{equation}
Then, the orthogonal score is $\psi^*(\beta , \theta ; w) = \psi(\beta , f ; w) + \phi(\beta , \theta ; w) $.
However, their discussion focused on  deriving $\phi$ based on specific statistic models. 
In contrast, we assume that the nuisance function $f_0$ arises from a general $M$-estimation.  
We derive the closed forms for $\phi$ and $\theta_0$ as given before Theorem \ref{NO2_score}.
Our approach is easy to implement and requires no assumptions about the model structure, making it applicable to a wide range of semiparametric models for deriving closed forms of Neyman-orthogonal scores for target parameters.

The idea is that we treat $\hat{f}_n$ as the working $M$-estimator of $f_0$, which satisfies $\hat{f}_n = \argmin_{f \in  \mathcal{F}_n} \mathbb{P}_n m_1(f; W) $ and exhibits the asymptotic functional normality.
Specifically, adapting the proof of Proposition \ref{prop0}, we have that under some regularity conditions, for any $h\in \mathcal{F}$,
$$-\sqrt{n}\mathbb{P} \{\partial^2_{ff} m_1(f_0;W)(\hat{f}_n-f_0) h\} = \sqrt{n}(\mathbb{P}_{n} - \mathbb{P}) \{\partial_f m_1(f_0;W)h\} +  o_p(1).$$
By choosing a direction $h_0 $ as shown in \eqref{weight}, such that
\begin{align*}
\sqrt{n}\mathbb{P}   \{\psi(\beta_0, \hat{f}_n; W) - \psi(\beta_0, f_0; W) \} =& \sqrt{n}\mathbb{P} \{\partial^2_{ff} m_1(f_0;W)(\hat{f}^M_n-f_0)  h_0\}   + o_p(1)\\
=& \sqrt{n}(\mathbb{P}_{n} - \mathbb{P}) \{\partial_f m_1(f_0;W)h_0\} + o_p(1),
\end{align*}
we can conclude that the specified influence function $\phi(\beta , \theta ; w) = \partial_f m_1(f ;w)h$ with $\theta_0= (f_0, h_0)$ is a valid first stage influence function.
\end{remark}

\begin{remark}[Comparison with 
\cite{chernozhukov2022debiased}]\label{RE4}
\cite{chernozhukov2022debiased} considered a specific setting where the target parameter is a linear functional of a nuisance regression function and showed that the Riesz representer (additional nuisance function) can be used to construct Neyman-orthogonal scores via the Riesz representation theorem.

Specifically, let $X$ and $Y$ denote the predictor and response variables, respectively, and let $w=(x,y)$.
Suppose that a score function for $\beta_0$ is $\psi(\beta,f;w) = \beta - \phi(f;w)$, and define an objective function for estimation of regression function $f_0$ as $m_1(f;w) = (f - y)^2$.
Under some mild conditions, Theorem \ref{NO2_score} implies that one Neyman-orthogonal score condition for $\beta_0$ is
\begin{align*}
0 = \beta_0 - \mathbb{P}\phi(f_0;W)  + \mathbb{P}  (\{ Y - f_0(X)\} E[\partial_f \phi(f_0;W) \mid X] ) ,
\end{align*}
which aligns with Definition 2.4 in \cite{chernozhukov2022debiased} when $ \mathbb{P} \phi(f;W) $ is a linear functional of $f$.
This shows that the Riesz representer may admit a closed form, $h_0(x) = E[\partial_f \phi(f_0;W) \mid X = x]$.
Therefore, compared to the approach relying on the Riesz representer, the proposed method is more adaptable to general problem settings and yields explicit results. These advantages facilitate the analysis of specific statistical models and significantly enhance the practical usability of Neyman-orthogonal scores.
\end{remark}

To clarify the relative advantages of our method, Table \ref{TAB_comp} summarizes the key differences between the proposed method and the alternative approaches discussed in Remarks \ref{RE2} and \ref{RE4}.

\begin{table}[htbp]                               \centering                                        \caption{Comparison of Neyman-orthogonal score construction frameworks: Proposed, Riesz Representer (\cite{chernozhukov2022debiased}), and Concentrating-out (\cite{10.1111/ectj.12097}).} \label{TAB_comp}     \renewcommand{\arraystretch}{1.3}                 \resizebox{0.9\textwidth}{!}{ 
\begin{tabular}{|p{2.6cm}| p{3.8cm} |p{3.8cm}| p{3.8cm}|}                                      
\hline 
\textbf{Feature} & \textbf{Proposed} & \textbf{Riesz Representer} & \textbf{Concentrating-out} 
\\
\hline                                            \textbf{Applicability} & General coupled and decoupled semiparametric models. & Models where target parameters are linear functionals of nuisance regression functions. & Jointly identified (fully coupled) models. \\ 
\hline  
\textbf{Nuisance structure} & General $M$-estimation and sequential systems. & Conditional expectation operators ($L_2$ spaces). & Models admitting tractable profiled minimizers. \\ 
\hline  
\textbf{Theoretical basis} & Constructive derivation via functional asymptotic theory. & Existential results via Riesz Representation Theorem. & Differentiation of concentrated objectives. \\ 
\hline  
\textbf{Orthogonality mechanism} & Explicit closed-form influence functions. & Debiasing via representer-weighted residuals. & Implicit profiling via parameter-dependent paths. \\ 
\hline  
\textbf{Interpretability} & Constructive closed-form representations for auxiliary nuisance functions. & "Black-box" weighting representers (additional nuisance functions). & Implicit parameter-dependent profiling. \\ 
\hline
\end{tabular}}
\end{table}

Further, to demonstrate the effectiveness of the proposed method, we apply it to a classical quantile treatment effect estimation problem in the observational study and provide a novel Neyman-orthogonal score function.

\begin{example}[Robust quantile treatment effect estimation]\label{ex2}
Let $Y$ and $D$ be the response variable and treatment assignment, respectively, and let $X$ consist of confounders.
Under the strong ignorability condition, \cite{firpo2007efficient} showed that the $\tau$th quantile of potential outcome $Y(1)$ can be identified as
\begin{align*}
\beta_0 = \arg\min_{\beta\in\mathbb{R}} \mathbb{P}\bigg\{ \frac{D}{p_0(X)} \cdot \rho_{\tau}(Y - \beta)\bigg\},
\end{align*}
where $p_0(x):=P(D=1\mid X=x)$ and $\rho_{\tau}(u)=u\{ \tau - I(u\leq0)\} $ with indicator function $I(\cdot)$.

Next, we apply Theorem \ref{NO2_score} to propose a robust estimation method for $\beta_0$.
The positivity assumption yields that there exists $f_0\in \mathcal{F}$ such that $p_0(x) = 1 /[ 1 + \exp\{ -f_0(x)\} ]$, which can be characterized by minimizing cross-entropy loss $m_1(f;w) = \log\ [ 1+\exp \{f(x) \}] -df(x) $ in population, where $w=(x,d,y)$.
Let $m(\beta,f;w) =  d \rho_{\tau}(y - \beta)  [ 1+\exp \{-f(x) \}] $, which is obviously non-differentiable at the point $\beta = y$ with respect to $\beta$.
By \cite{10.1214/10-AOS827} and Fubini's Theorem,
we can demonstrate that $\psi(\beta,f ;w)=d\{I(y-\beta\leq0)-\tau\} [ 1+\exp \{-f(x) \}]$ is a score function for $\beta_0$ with nuisance function $f_0$.
Then, combining $\partial_{f}  \psi(\beta_0, f_0;w) = -d\{I(y-\beta_0\leq0)-\tau\}\exp \{ -f_0(x)\} $ and $\partial^2_{ff} m_{1}( f_0;w) = \exp \{f_0(x)\} /[ 1+\exp \{f_0(x) \}])^2$, we have
\begin{align*}
h_0(x) = -\frac{E[\partial_{f} \psi(\beta_0, f_0;W)\mid X=x]}{E[\partial^2_{ff} m_{1}( f_0;W)\mid X=x]}= \frac{E[D\{I(Y\leq \beta_0) - \tau\} \mid X=x]}{ 1 / [1 + \exp\{-f_0(x)\}] ^2 }.
\end{align*}
Hence, by Theorem \ref{NO2_score}, the Neyman-orthogonal score function for $\beta_0$ with nuisance functions $\theta_0= (f_0,h_0)$ is
\begin{align*}
\psi^*(\beta, f_0,h_0;W) 
= &\frac{D [1+\exp\{f_0(X) \} ]}{\exp\{f_0(X)\}}\{ I(Y-\beta_0\leq 0) - \tau\} + \biggl[ \frac{\exp\{f_0(X)\}}{1+\exp\{ f_0(X)\}}- D \biggr]h_0.
\end{align*} 
Note that $ \partial_{\beta}\mathbb{P} \psi^*(\beta_0, f_0,h_0;W) = f_{Y(1)}(\beta_0)$, where $f_{Y(1)}$ denotes the probability density function of potential outcome $Y(1)$. 
Under conditions in Theorem \ref{NO2_score}, $\hat{\beta}_n^{R}$ obtained via Algorithm \ref{alg:cross_fitting} satisfies
\begin{align*}
&-\sqrt{n}f_{Y(1)}(\beta_0)(\hat{\beta}_n^{R}  - \beta_0)+ o_p(1) \\
=& \frac{1}{ \sqrt{n}}\sum_{i=1}^{n} \biggl\{\frac{D_i}{p_0(X_i)}\{ I(Y_i\leq\beta_0) -\tau \}+\frac{p_0(X_i)-D_i}{p_0^2(X_i)}E[D\{I(Y\leq \beta_0) -\tau\} \mid X_i ] - \beta_0 \biggr\}.
\end{align*}

Different from the theory developed by \cite{firpo2007efficient}, the asymptotic normality of the robust estimator $\hat{\beta}_n^{R}$ does not depend on the methods used to construct estimators of nuisance functions $h_0$ and $p_0$, provided they are $o_p(n^{-1/4})$-consistent. 
This flexibility allows us to apply powerful machine learning approaches to perform nonparametric estimation tasks and implement statistical inference with theoretical guarantees.
Furthermore, the asymptotic variance of our robust estimator matches that of the estimator  in \cite{firpo2007efficient}, demonstrating that $\hat{\beta}_n^{R}$ successfully achieves the semiparametric efficiency bound.

\end{example}

\subsection{Semiparametric model with multiple nuisance functions}\label{SI1}
This subsection extends the method of Subsection \ref{SI2} to a general setting where there are multiple nuisance functions.
The coupled model in Subsection \ref{S1}, where the parameter of interest $\beta_0$ and the nuisance function $f_0$ are jointly identified, can be generalized by treating $f_0$ as a vector of nuisance functions.
The corresponding methods and theoretical results, as established in Theorem \ref{NO_score}, can be applicable.

On the other hand, when multiple nuisance functions must be identified sequentially, the methodology of Subsection \ref{S2} needs to be extended carefully.
The sequential dependence between finite-dimensional nuisance functions complicates the task of establishing asymptotic normality for functionals of their estimators.
In this subsection, we provide a comprehensive approach for constructing the influence function and identifying additional nuisance functions, leveraging asymptotic functional normality theory under a two-stage estimation process.

Let $\Theta$ denote the set of nuisance functions on $\mathcal{X}$ that are uniformly bounded by $B$.
The target parameter, $\beta_0 \in \mathbb{R}$,  is identified by the score condition 
\begin{align}\label{ori_No3}
\mathbb{P} \psi(\beta_0,\mu_0, f_0; W) = 0,
\end{align}
where score function $\psi$ is known and the nuisance functions $(\mu_0, f_0) \in \Theta \times  \mathcal{F}$ are defined through sequential optimization problems involving criterion functions $m_1$ and $m_2$ in the population:
\begin{equation}\label{nuisance}
\mu_0 = \argmin_{\mu\in\Theta} \mathbb{P} m_2(\mu, f_0; W), \quad f_0 = \argmin_{f\in \mathcal{F}} \mathbb{P}m_1(f; W).
\end{equation}
A concrete example of this structure is provided by \cite{10.1093/jrsssb/qkae024}, who considered a regression model with multiple weak invalid instruments. 

We construct the working estimators of $\mu_0$, $f_0$ and $\beta_0$ via sample splitting.
Randomly partition the sample $S$ into two subsets, $S_1 = \{W_i\}_{i=1}^{n_1}$ and $S_2 = \{W_i\}_{i=n_1+1}^n$, with a proportion $n_1/n = p$, where $0 < p < 1$. The estimation proceeds as follows:
\begin{itemize}
\item \textit{Stage 1: Estimation of \( f_0 \).} Using the subsample \( S_1 \), estimate the nonparametric nuisance function via empirical risk minimization:
\begin{equation} \label{2.2}
\hat{f}_{n_1} = \argmin_{f \in  \mathcal{F}_n}\mathbb{P}_{n_1} m_1(f; W),
\end{equation}
where $ \mathcal{F}_n \subseteq  \mathcal{F}$ is a sieve space.

\item \textit{Stage 2: Estimation of \( \mu_0 \).} Using the independent subsample \( S_2 \) and the first-stage estimate \( \hat{f}_{n_1} \), estimate the target parameter:
\begin{equation} \label{2.3}
\hat{\mu}_{n_2} = \argmin_{\mu \in \Theta_n} \mathbb{P}_{n_2} m_2(\mu,\hat{f}_{n_1};W),
\end{equation}
where $\Theta_n \subseteq \Theta$ is a sieve space and \( n_2 = n - n_1 \).
\end{itemize}
Define $\hat{\beta}_n $ as the estimator that satisfies
$n^{-1} \sum_{i=1}^n \psi(\hat{\beta}_n ,\hat{\mu}_{n_2}, \hat{f}_{n_1}; W_i)=0.$

Next, we characterize the influence function of $\hat{\beta}_n$ by leveraging the asymptotic functional normality of the nuisance parameter estimators $\hat{\mu}_{n_2}$ and $\hat{f}_{n_1}$.  
We propose the following proposition.
\begin{proposition}\label{prop1}
Under some conditions, for any $h\in\Theta$, we have
\begin{align*}
-\sqrt{n_2} \ \partial^2_{\mu \mu} \mathbb P m_2(\mu_0,f_0;W)[h] [\hat{\mu}_{n_2}-\mu_0] 
=&\sqrt{n_2}\ \partial^2_{\mu f} \mathbb P m_2(\mu_0,f_0;W)[h][\hat{f}_{n_1}-f_0] \\
&+\sqrt{n_2}\ (\mathbb{P}_{n_2}-\mathbb{P})\{\partial_{\mu} m_2(\mu_0,f_0;W) h\} +o_p(1).
\end{align*}
\end{proposition}
We emphasize that the regularity conditions for Proposition \ref{prop1} are used to derive the asymptotic functional normality of the working estimator $\hat{\mu}_{n_2}$.
These conditions are independent of the establishment of Neyman orthogonality.
Therefore, we omit the formal statement of these conditions here; a comprehensive list of regularity conditions and the detailed proof of Proposition \ref{prop1} are provided in the Appendix.

Subsequently, we derive the influence function of $\hat{\beta}_n$ and show that it serves as a valid Neyman-orthogonal score function for $\beta_0$. 
By Taylor expansion, we obtain
\begin{align*}
&-\sqrt{n}(\hat{\beta}_n  - \beta_0)\  \partial_{\beta} \mathbb{P}  \psi(\beta_0,\mu_0, f_0;W)   \nonumber \\
= & \underbrace{\sqrt{n}\ \partial_{f} \mathbb{P}   \psi(\beta_0,\mu_0,f_0;W)[\hat{f}_{n_1}-f_0] }_{A} +\underbrace{\sqrt{n}\ \partial_{\mu} \mathbb{P}   \psi(\beta_0,\mu_0,f_0;W)[\hat{\mu}_{n_2}-\mu_0]}_{B}\\
& +\sqrt{n} \  (\mathbb{P}_n - \mathbb{P}) \psi(\beta_0,\mu_0, f_0;W) + \text{residual}.
\end{align*}
We explore the influence functions of $A$ and $B$, respectively.
Under some conditions, adapting the proof of Proposition \ref{prop0}, we have that for arbitrary direction $h_1\in \mathcal{F}$,
$$-\sqrt{n}\ \partial^2_{ff} \mathbb{P} m_1(f_0;W)[\hat{f}_{n_1} - f_0][h_1] = \sqrt{n}\ (\mathbb{P}_{n_1} - \mathbb{P}) \{\partial_f m_1(f_0;W)h_1\} +  o_p(1).$$
Hence, choosing well-defined
$h_{10}(x) = -E[\partial_{f} \psi(\beta_0,\mu_0,f_0;W)\mid X=x] / E[\partial^2_{ff} m_1(f_0;W)\mid X=x]$,
we have
\begin{equation}
A = \sqrt{n}\ (\mathbb{P}_{n_1} - \mathbb{P}) \{\partial_f m_1(f_0;W)h_{10}\} +  o_p(1).\nonumber
\end{equation}
By Proposition \ref{prop1}, choosing

$h_{20}(x) = -E[\partial_{\mu} \psi(\beta_0,\mu_0,f_0;W)\mid X=x] / E[\partial^2_{\mu \mu} m_2(\mu_0,f_0;W)\mid X=x] $
and

$h_{30}(x) = -E[\partial^2_{\mu f} m_2(\mu_0,f_0;W)h_{20}\mid X=x] / E[\partial^2_{ff} m_1(f_0;W)\mid X=x]$, we conclude that
\begin{equation}
B = \sqrt{n}(\mathbb{P}_{n_1} - \mathbb{P}) \{\partial_f m_1(f_0;W)h_{30}\} + \sqrt{n}(\mathbb{P}_{n_2} - \mathbb{P}) \{\partial_{\mu} m_2(\mu_0,f_0;W)h_{20} \}+  o_p(1). \nonumber
\end{equation}
Therefore, we have
\begin{align*} 
&-\sqrt{n} \ (\hat{\beta}_n - \beta_0) \partial_{\beta} \mathbb{P}  \psi(\beta_0,\mu_0, f_0;W)\nonumber \\
= &\sqrt{n} \ (\mathbb{P}_{n_1} - \mathbb{P}) \{\partial_f m_1(f_0; W)(h_{10} + h_{30})\} + \sqrt{n} (\mathbb{P}_{n} - \mathbb{P}) \psi(\beta_0,\mu_0, f_0; W) \nonumber \\
& + \sqrt{n}\  (\mathbb{P}_{n_2} - \mathbb{P}) \{\partial_\mu m_2(\mu_0, f_0; W)h_{20}\} + o_p(1).
\end{align*}

Let 
\begin{align}
h_{10}(x) &= -E[\partial_f \psi(\beta_0,\mu_0, f_0; W) \mid X=x] / E[\partial^2_{ff} m_1(f_0; W) \mid X = x], \label{h1}\\
h_{20}(x) &= -E[\partial_\mu \psi(\beta_0,\mu_0, f_0; W) \mid X = x] / E[\partial^2_{\mu \mu} m_2(\mu_0, f_0; W) \mid X= x], \label{h2}\\
h_{30}(x) &= -E[\partial^2_{\mu f} m_2(\mu_0, f_0; W) h_{20}(X) \mid X = x] / E[\partial^2_{ff} m_1(f_0; W) \mid X= x] \label{h3}.
\end{align}

\begin{theorem}[Neyman Orthogonality in the semiparametric model with multiple nuisance functions] \label{co1}
Suppose the target parameter $\beta_0$ is identified via a score condition in \eqref{ori_No3}, and nuisance functions $(f_0, \mu_0)$ are independently characterized by \eqref{nuisance}. Assume the following conditions hold:

\begin{itemize}
\item[(a)] The score function $\psi$ in \eqref{ori_No3} is twice continuously differentiable, and the objective functions $m_1$ and $m_2$ in \eqref{nuisance} are thrice continuously differentiable with respect to their arguments.
\item[(b)] The functions $h_{10}$, $h_{20}$, and $h_{30}$ in \eqref{h1}--\eqref{h3} are well-defined.
\item[(c)] $\partial_{f} \mathbb{P} m_1(f_0; W)[h_1] = 0$ and $\partial_{\mu} \mathbb{P} m_2(\mu_0, f_0; W)[h_2] = 0$ for any directions $h_1 \in \mathcal{F}$ and $h_2 \in \Theta$.
\end{itemize}

Then, the following properties hold: 
\begin{itemize}
\item[(i)] The function 
\begin{align}\label{def_NO3}
&\psi^*(\beta ,\mu ,f ,h_1,h_2,h_3;w)
=\psi(\beta , \mu , f ; w) + \partial_f m_1(f ; w)(h_1 + h_3) + \partial_\mu m_2(\mu , f ; w)h_2 
\end{align}
serves as a valid score function for $\beta_0$ with nuisance functions $\theta_0 = (f_0, \mu_0, h_{10}, h_{20}, h_{30})$. 
Furthermore, it obeys the Neyman-orthogonal condition at the $(\beta_0, \theta_0)$ with respect to $\theta_0$.

\item[(ii)] Assume that estimators $(\hat{\theta}_{n_1},\hat{\theta}_{n_2},\hat{\beta}_{n_1},\hat{\beta}_{n_2})$ constructed in Algorithm \ref{alg:cross_fitting} based on $\psi^*$ in \eqref{def_NO3} are $o_p(n^{-1/4})$-consistent. 
Then, the output of Algorithm \ref{alg:cross_fitting}, denoted by $\hat{\beta}_n^{R}$, has the asymptotic expansion
$$-\sqrt{n}(\hat{\beta}_n^{R} - \beta_0)\partial_{\beta}\mathbb{P}\psi^*(\beta_0, \theta_0; W) = \sqrt{n}(\mathbb{P}_n-\mathbb{P})\psi^*(\beta_0, \theta_0; W) + o_p(1).$$
\end{itemize} 
\end{theorem}

As proof of Theorem \ref{co1} is similar to that of Theorem \ref{NO_score}, it is omitted.

The asymptotic normality of the robust estimator $\hat{\beta}_n^R$ is guaranteed by the existence of $o_p(n^{-1/4})$-consistent estimator for the nuisance function $\theta_0$.
This condition is relatively mild in practice and can be achieved using powerful sieve approximation tools, such as deep neural networks or other machine learning methods.
To the best of our knowledge, this represents a novel approach for systematically achieving orthogonality within semiparametric models.

\section{Application to causal inference with a binary instrumental variable} \label{Sec4} 

In this section, we apply the proposed method to construct a Neyman-orthogonal score function for a causal effect of interest, thereby ensuring locally robust estimation and inference in the presence of a binary instrumental variable. 

The instrumental propensity score weighting method proposed by \cite{abadie2003semiparametric} is widely used to identify and estimate causal effects with a binary IV (\cite{qiu2021inference,sloczynski2025abadie}). 
In particular, this approach reweights the observed data distribution to recover the distribution of ``compliers''---those whose treatment status responds to the instrumental assignment, as defined by \cite{angrist1996identification}.
The corresponding estimation procedure relies on the instrumental propensity score, which is usually treated as an infinite-dimensional nuisance parameter. 
However, this classical weighting approach lacks robustness, making the causal effect estimator highly sensitive to first-step estimation errors in the nuisance function. 
To address this issue, our goal is to augment the weighting method by introducing a Neyman-orthogonal correction to the standard score functions.

\subsection{Notation and assumptions}\label{EOL} 
The data consist of $n$ observations on a response variable $Y$, a binary treatment indicator $D$, a binary instrument $Z$ and a $p$-dimensional covariate vector $X $.
Let $Y(z, d) $ be the potential outcome under $D=d$ and $Z=z$ for $d,z\in\{0,1\}$.
Denote $Y(0)$ and $Y(1)$ as the potential outcomes when $D=0$ and $D=1$, respectively.
Then, the observed response variable is $Y=Y(1)D+Y(0)(1-D)$.
Denote $D(0)$ and $D(1)$ as the potential treatment when $Z=0$ and $Z=1$, respectively.
The observed treatment assignment is $D=D(1)Z+D(0)(1-Z)$.
The population naturally partitions into four subgroups based on potential treatment behaviors: always-takers ($D(0) = D(1) = 1$), compliers ($D(1) > D(0)$), never-takers ($D(0) = D(1) = 0$), and defiers ($D(0) > D(1)$).
Let $W = (X, D, Y, Z)$ and $w = (x, d, y, z)$ denote a random vector and a realization, respectively.

Our goal is to identify and estimate the causal effect for the complier subpopulation under the following four standard assumptions (\cite{abadie2003semiparametric}):

(C1) Instrument independence: Conditional on $X$, the random vector

$(Y(0,0),Y(0,1),Y(1,0),Y(1,1),D(0),D(1))$ is independent of $Z$.

(C2) Exclusion restriction: $P(Y(1,d)=Y(0,d)\mid X)=1$ for $d\in\{0,1\}$.

(C3) First-stage relevance: $0<P(Z=1\mid X)<1$ and $P(D(1)=1\mid X)>P(D(0)=1\mid X)$.

(C4) Monotonicity: $P(D(1)\geq D(0)\mid X)=1$.

Assumption (C1) establishes the ignorability condition, ensuring that the instrumental variable $Z$ can be treated as effectively randomly assigned when conditioned on $X$.
(C2) enforces the exclusion restriction by specifying that potential outcomes depend solely on the treatment indicator $D$, allowing the potential outcomes $Y(D,Z)$ to be simplified to $Y(D)$. 
(C3) implies that variables $D$ and $Z$ remain correlated conditioning on $X$, while also satisfying the positivity for the propensity score.
(C4) excludes the existence of defiers in the population.

We focus on the following causal effect of interest:
$$\beta_0 := E[\{Y(1)-Y(0)\} \{D(1)>D(0)\}],$$
which captures the local average treatment effect (LATE, $E\big[Y(1) - Y(0) \mid D(1) > D(0)\big]$) for the complier subpopulation, weighted by the probability $P(D(1)>D(0))$ (\cite{angrist1995two}).

The local average response functions (LARFs) under treatment assignment $t\in\{0,1\}$ are defined as
\begin{align}\label{def_LARF}
\mu_0^{(t)}(x)=E[Y\mid X=x,D=t,D(1)> D(0)].
\end{align} 

Let
$g_0(X) = P(Z=1 \mid X)$
denote the instrumental propensity score. 
Under Assumption (C3), the corresponding log-odds $f_0(X)= \log\{g_0(X) / (1 - g_0(X))\}$ is well-defined. 
Furthermore, by employing the standard cross-entropy criterion within a nuisance function space $\mathcal{F}$, $f_0$ can be characterized as the minimizer of the following expected loss:
\begin{equation}\label{propensity}
f_0 = \argmin_{f\in\mathcal{F}} \mathbb{P}\big[\log\big\{1+e^{f(X)}\big\} - Zf(X)\big].
\end{equation}
The detailed specification of $\mathcal{F}$ will be discussed later in Remark \ref{re7}.
Subsequently, the instrumental propensity score can be recovered via the transformation, $g_0 = \operatorname{expit}(f_0)$, where $\operatorname{expit}(x) = e^x / (1 + e^x)$ is the sigmoid function.

\subsection{Identification of the causal effect}
In this section, we present the identification and estimation strategy for $\beta_0$. 

As established by \cite{abadie2003semiparametric}, these functions are identifiable under Assumptions (C1)–(C4). 
Therefore, the identification of $\beta_0$ can be achieved through the LARFs in \eqref{def_LARF}, expressed as the following calculation:
\begin{align}\label{ori_score}
\beta_0 &= E\{E[Y(1) - Y(0) \mid X, D(1)> D(0)]\mid D(1) > D(0)\} \cdot P(D(1) > D(0)) \nonumber \\ 
&=E [\mu_0^{(1)}(X)-\mu_0^{(0)}(X)\mid D(1)>D(0)] \cdot P(D(1) > D(0)).
\end{align}
Adapting Abadie's kappa weighting (\cite{abadie2003semiparametric}), we have the following proposition:
\begin{proposition}\label{prop2}
Define 
\begin{align*}
&\kappa^{(0)}(X,D,Z;g)=(1-D)\frac{(1-Z)-\{1-g(X)\}}{\{1-g(X)\}g(X)},
\\
&\kappa^{(1)}(X,D,Z;g)=D\frac{Z-g(X)}{\{1-g(X)\}g(X)}.
\end{align*}
Let $\kappa^{(1)}_0 = \kappa^{(1)}(X,D,Z;g_0)$ and $\kappa^{(0)}_0 = \kappa^{(0)}(X,D,Z;g_0)$, where $g_0$ is the instrumental propensity score.

Let $h$ be any measurable real function of $(Y,X)$, such that $E|h(Y, X)| < \infty$.
Under Assumptions (C1)-(C4), 
\begin{align*}
& E[h(Y(0),X) \mid D(1)>D(0)] = \frac{1}{P(D(1)>D(0))} E[\kappa^{(0)}_0 h(Y,X)] , \\
& E[h(Y(1),X) \mid D(1)>D(0)] = \frac{1}{P(D(1)>D(0))} E[\kappa^{(1)}_0 h(Y,X)].
\end{align*}
\end{proposition}  
\noindent Based on the kappa weighting in Proposition \ref{prop2}, we obtain two fundamental results for the identification of the LARFs and the causal effect $\beta_0$.
\begin{itemize}
\item[(i)] For $t\in\{0,1\}$ and the nuisance function space $\Theta$ that contains $\mu_0^{(0)}$ and $\mu_0^{(1)}$, LARFs can be characterized by the nonparametric re-weighted least squares criterion:
\begin{align}\label{LARF}
\mu_0^{(t)} =&\argmin_{\mu\in\Theta} E[\{Y(t)-\mu(X)\}^2\mid D(1)>D(0)] 
= \argmin_{\mu\in \Theta} \mathbb{P}\kappa^{(t)}_0\{Y-\mu(X)\}^2.
\end{align}

\item[(ii)] A valid score function of $\beta_0$ derived from \eqref{ori_score} is 
\begin{align}\label{ori}
\psi(\beta,g,\mu^{(1)},\mu^{(0)};w) =  \kappa^{(1)}(x,d,z;g)\mu^{(1)} - \kappa^{(0)}(x,d,z;g)\mu^{(0)} - \beta
\end{align}
with nuisance functions $(g_0, \mu_0^{(1)}, \mu_0^{(0)})$.
\end{itemize} 

\subsection{Locally robust estimation and inference for $\beta_0$}  \label{LS}
This section provides a locally robust estimation and inference method for $\beta_0$ based on the Neyman-orthogonal score function.

Motivated by \eqref{LARF} and \eqref{ori}, we treat the propensity score $g_0$ as the first-stage nuisance function. 
The semiparametric structure defined as \eqref{propensity}--\eqref{ori} aligns with the general approach discussed in the Subsection \ref{SI1}.
Therefore, we apply Theorem \ref{co1} to establish a Neyman-orthogonal score function for $\beta_0$ as follows:
\begin{align}\label{robust1}
\psi^*(\beta , f , h ;w)
= & \{\kappa^{(1)}(x,d,z; \operatorname{expit}(f)) - \kappa^{(0)}(x,d,z; \operatorname{expit}(f))\}y \nonumber \\
& - \frac{\operatorname{expit}(f) -z}{\operatorname{expit}(f) \{1- \operatorname{expit}(f) \}}h - \beta,
\end{align}
where the nuisance vector $\theta_0 = (f_0, h_0)$ involves the nuisance function $h_0$, defined as:
\begin{equation}\label{h_0}
h_0(x)=E\bigg[Y\bigg\{ \frac{e^{2f_0(X)}-1}{e^{f_0(X)}}Z - e^{f_0(X)} \bigg\}\bigg|X=x\bigg].  
\end{equation} 

Before continuing the discussion, we make a remark that is important for understanding the proposed Neyman-orthogonal score function.

\begin{remark}
The specification of the score function for $\beta_0$ is not unique.
In particular, different from \eqref{ori},  Proposition \ref{prop2} implies another valid score function for $\beta_0$:
\begin{align}\label{moment}
\psi^{Moment}(\beta,g;w) = \{\kappa^{(1)}(x,d,z;g) - \kappa^{(0)}(x,d,z;g)\} y - \beta.
\end{align}
In the presence of a single nuisance function $g_0$ characterized by \eqref{propensity}, we can apply Theorem \ref{NO2_score} to construct a Neyman-orthogonal score function for $\beta_0$. 
Algebraic calculation shows that the Neyman-orthogonal augmentation of \eqref{moment} coincides with $\psi^*$, which is shown as \eqref{robust1} with nuisance functions $\theta_0$.
This implies that different original score functions may be represented by the common Neyman-orthogonal score function.
Furthermore, since these proposed Neyman-orthogonal score functions serve as the influence functions for respective estimators, under some conditions, the estimators of $\beta_0$ derived from $\psi$, $\psi^{Moment}$, and $\psi^*$ are asymptotically equivalent and share one asymptotic variance.
\end{remark}

Denote the robust estimator of $\beta_0$ by $\hat{\beta}^R_n$, which is obtained via implementing Algorithm \ref{alg:cross_fitting} with the score function $\psi^*$ in \eqref{robust1} and nuisance functions, $ f_0 $ in \eqref{propensity} and $h_0$ in \eqref{h_0}. 

\begin{theorem}\label{theo:Causal_inf}
(i). The score function $\psi^*$ obeys the Neyman-orthogonal condition at $(\beta_0,f_0,h_0)$ with respect to $\theta_0$.

(ii). If nuisance and target parameter estimators constructed in Algorithm \ref{alg:cross_fitting} with $\psi^*$, $\hat{\theta}_{n_1} = (\hat{f}_{n_1}, \hat{h}_{n_1}) $, $\hat{\theta}_{n_2}= (\hat{f}_{n_2}, \hat{h}_{n_2})  $, $\hat{\beta}_{n_1}$ and $\hat{\beta}_{n_2}$ are $o_p(n^{-1/4})$-consistent, then the estimator $\hat{\beta}_n^{R} $ is asymptotically normal:
$$\sqrt{n}(\hat{\beta}_n^R  - \beta_0)
\rightarrow_{d} N(0,\sigma^2),$$
as $n\rightarrow\infty$, where
\begin{equation}
\sigma^2= \  E \bigg[\kappa^{(1)}_0Y- \kappa^{(0)}_0Y - \frac{g_0(X)-Z}{g_0(X)\{ 1-g_0(X)\} }h_0(X) - \beta_0\bigg]^2. \nonumber
\end{equation}
\end{theorem}

The proof of Theorem \ref{theo:Causal_inf} is given in the Appendix.

\begin{remark}[Convergence rates of estimators]\label{re7}
Given that all estimators are $o_p(n^{-1/4})$-consistent, the asymptotic normality of the robust estimator can be achieved and does not depend on the method used to construct the nuisance parameter estimators.
When the covariate vector is $p$-dimensional, taking $(\hat{\theta}_{n_1},\hat{\beta}_{n_2})$ obtained in Algorithm \ref{alg:cross_fitting} for instance, we provide theoretical guidelines for constructing these $o_p(n^{-1/4})$-consistent estimators.

First, the $o_p(n^{-1/4})$-consistency of the log-odds estimator, $\hat{f}_{n_1}$, can be attained using ReLU deep neural networks (DNNs). 
We adopt the ReLU activation function, defined as $\rho(v) = \max(0, v)$. Assuming the hidden layers of the DNN have a constant width $H$, the DNN function $h: \mathbb{R}^{p} \to \mathbb{R}$ can be parameterized as:
\begin{align*}
h(x) = s_{\mathcal{D}} \rho \big( \cdots \rho(s_1 \rho(s_0 x + b_0) + b_1) \cdots \big) + b_{\mathcal{D}}.
\end{align*}
Here, $\mathcal{D}$ represents the number of hidden layers; $s_0 \in \mathbb{R}^{H \times p}$ and $b_0 \in \mathbb{R}^H$ map the $p$-dimensional input to the first hidden layer; for $l = 1, \dots, \mathcal{D}-1$, $s_l \in \mathbb{R}^{H \times H}$ and $b_l \in \mathbb{R}^H$ are the weight matrices and bias vectors connecting the hidden layers; and $s_{\mathcal{D}} \in \mathbb{R}^{1 \times H}$ and $b_{\mathcal{D}} \in \mathbb{R}^1$ yield the scalar output. 

Recent theoretical advances (\cite{schmidt2020nonparametric, farrell2021deep, jiao2023deep}) established the convergence rates of DNNs in general nonparametric M-estimation problems. Let $P_X$ denote the probability measure of $X$. In particular, assuming that $f_0$ belongs to a $\gamma$-H\"older ball with $\gamma \in\mathbb{N}_{+}$, and adapting Corollary 1 of \cite{farrell2021deep}, we obtain the following error bound under mild regularity conditions:
\begin{align*}
\int | \hat{f}_{n_1}(x) - f_0(x) |^2 \,dP_X(x) = O_p\bigg(n^{-\frac{2\gamma}{2\gamma+p}}\log^4 n\bigg).
\end{align*}
The above rate holds provided the network architecture satisfies $\mathcal{W}=O( n^{p/(2\gamma + p)}\log n)$ (where $\mathcal{W}$ is the total number of parameters, i.e., $\mathcal{W}=(p+1)H+(\mathcal{D}-1)(H^2 + H) + H + 1$) and $\mathcal{D}=O(\log n)$. 
Therefore, if $\gamma > p/2$, then $\hat{f}_{n_1}$ is $o_p(n^{-1/4})$-consistent.

Second, the estimator $\hat{h}_{n_1}$ is constructed via a two-stage approach: (i) first constructing the individual-level target variable $Y\{(e^{f_0(X)}-e^{-f_0(X)})Z - e^{f_0(X)}\}$; and (ii) regressing this target variable on the $X$. The overall convergence rate of $\hat{h}_{n_1}$ is governed by the slower rate between the estimators for $E[Y(Z-1) \mid X=x]$ and $f_0$. Therefore, following the aforementioned arguments, the $o_p(n^{-1/4})$-consistency of $\hat{h}_{n_1}$ can similarly be achieved utilizing DNN methods.

Third, the estimator $\hat{\beta}_{n_2}$ is obtained by solving the estimating equation with respect to $\beta$: $\mathbb{P}_{n_2} \psi^*(\beta ,\hat{f}_{n_1}, \hat{h}_{n_1};W) =0$. As $\beta$ is a scalar parameter and the score function $\psi^*$ satisfies standard smoothness conditions, the $o_p(n^{-1/4})$-consistency of $\hat{\beta}_{n_2}$ is directly derived from that of the nuisance vector $\hat{\theta}_{n_1}$.
\end{remark}

\section{Simulation studies}\label{simulation}
In this section, we evaluate the finite-sample properties of the proposed robust estimator of $\beta_0$ as defined in Section \ref{Sec4} via simulations.
All numerical calculations are implemented in \texttt{Python}.

The observed data $\{W_i\}_{i=1}^n=\{(X_i, D_i, Y_i, Z_i)\}_{i=1}^n$ were generated with sample sizes $n=500$, $1000$, $2000$ and $4000$ through the following data-generating process.
\begin{itemize}
\item For response $Y$ among compliers, we used the model
\begin{equation*}
Y(D)=D\mu^{(1)}_0(X)+(1-D)\mu^{(0)}_0(X)+\epsilon,
\end{equation*}
where $\epsilon$ is independent of $(X, D)$ and follows a standard normal distribution.
For always takers, we set $E[Y\mid X,D,D(1)=D(0)=1]=X_1+X_2+X_3+X_4+2D$,
and for never takers, we set $E[Y\mid X,D,D(1)=D(0)=0]=0.6X_1+0.8X_2+X_3+1.2X_4-2D$.

\item The $p$-dimensional covariate vector $X$ followed a normal distribution $N(\boldsymbol0, I_p)$ with each component truncated to the interval $[-1,1]$, where $p=4,\ 10$.

\item Given the covariate vector $X$ and the instrumental propensity score function $\operatorname{expit}\{f_0(x)\}$, we generated the binary instrumental variable $Z$ with $P(Z=1\mid X)=\operatorname{expit}\{f_0(X)\}$.

\item The treatment indicator $D$ was set as $D=I(U=1)+ZI( U=2)$.
Here, the latent variable $U$ was sampled from $\{1,2,3\}$ with probabilities $0.2, 0.6, 0.2$ respectively. This ensures that $20\%$ of the subjects are always takers ($U=1$), $60\%$ are compliers ($U=2$), and $20\%$ are never takers ($U=3$).
\end{itemize}

We considered the instrumental propensity score function setting:
$f_0(x)=x_1^2x_2^3+\log(x_2x_3+4)-\exp\{x_3x_4/2\}-0.5$
and two LARF scenarios:
\begin{itemize}
\item[S1] $\mu_0^{(t)}(x)=\cos(\pi x_1x_2)+x_1x_2x_3^3+\exp(x_2x_3-1)+\log(3+x_3x_4)+3t$ with $\beta_0 =1.8$;
\item[S2] $\mu_0^{(t)}(x)=\sin(\pi x_1x_2/2)+\log(x_2x_3+1.5)+\exp(x_3x_4/2)+3t$ with $\beta_0=1.8$.
\end{itemize}

Here we evaluated the finite-sample performance of three methods:

(i) Robust (R): The proposed method, as formalized in Theorem \ref{theo:Causal_inf};

(ii) Moment (M): The naive moment-based estimator, $E[(\kappa^{(1)} - \kappa^{(0)})Y] $;

(iii) Regression imputation (Reg): The regression imputation estimator, $ E[\kappa^{(1)}\mu^{(1)}(X) - \kappa^{(0)}\mu^{(0)}(X)]  $.

The 2-fold cross-fitting procedure outlined in Algorithm \ref{alg:cross_fitting} was applied. 
All evaluated methods rely on the estimation of the log-odds $f_0$. 
To ensure consistency, this was estimated via nonparametric logistic regression using ReLU deep neural networks (DNNs), as detailed in Remark \ref{re7}. 
With the exception of the moment-based method (M), which does not require additional nuisance functions, we employed both parametric and nonparametric approaches to estimate the remaining nuisance parameters, $h_0$ and $(\mu^{(0)}_0, \mu^{(1)}_0)$:
\begin{itemize}
\item For the robust method (R), the nuisance function $h_0$ was estimated using either a linear model (R-LR) or a DNN (R-NP), both trained by minimizing the ordinary least squares loss.
\item For regression imputation (Reg), the conditional means $(\mu^{(0)}_0, \mu^{(1)}_0)$ were estimated using either linear models (Reg-LR) or DNNs (Reg-NP), both trained by minimizing the kappa weighting least squares loss.
\end{itemize} 
For all nonparametric estimation tasks involving DNNs (i.e., estimating $f_0, h_0, \mu^{(0)}_0$, and $\mu^{(1)}_0$), we adopted the fully connected network construction described in \cite{farrell2021deep} and Remark \ref{re7}. 
In particular, the network architecture was fixed at $\mathcal{D}=4$ hidden layers, each with a constant width of $H=80$ and utilizing the ReLU activation function. 
The models were trained via the Adam optimizer (\cite{kingma2014adam}) with the learning rate held constant at $0.001$. 

Estimation performance was assessed over $r=500$ replications using the average absolute bias and the scaled mean squared error (SMSE), defined as:
\begin{align*}
\text{Bias} =  \bigg| \frac{1}{r} \sum_{j=1}^{r}(\hat{\beta}_j -\beta_0)\bigg|, \
\text{SMSE} =  \frac{\sqrt{n}}{r} \sum_{j=1}^{r}(\hat{\beta}_j -\beta_0)^2 ,
\end{align*}
where $\hat{\beta}_j$ denotes the estimate from the $j $-th replication.


Figures~\ref{fig2} and \ref{fig3} summarize the performance metrics, Bias and SMSE, of the five methods across different sample sizes and covariate dimensions for the two scenarios. 
Across all evaluated scenarios and covariate dimensions, increasing the sample size from $500$ to $4000$ consistently reduces the SMSE for all methods, reflecting decreased estimation variance and improved estimator precision.
Both robust methods, R-NP and R-LR, consistently exhibit the lowest SMSE across all sample sizes, scenarios, and covariate dimensions.
Their Biases are also low, though not always the smallest, indicating the strong robustness and superior accuracy.
In contrast, the moment method yields higher SMSE than the robust approaches but outperforms the regression imputation methods, which are highly sensitive to model misspecification.
Increasing the covariate dimension from $p=4$ to $10$ increases SMSE for all methods, particularly in smaller samples.
Nevertheless, the proposed robust estimators maintain their performance advantage, demonstrating resilience to increased dimensionality.
The similarity in Bias and SMSE between R-NP and R-LR suggests that the robust method is largely insensitive to whether $h_0$ is estimated parametrically or nonparametrically.
This robustness is particularly important, as R-LR remains reliable even under potential model misspecification, whereas the regression imputation methods need to be accurately specified to perform well.

Tables~\ref{lab1} and \ref{lab2} report 95\% coverage probabilities based on the estimated standard errors.
Different from other methods, the proposed estimators, R-NP and R-LR achieve coverage close to the nominal 95\%, providing robust inference despite potential misspecification of $h_0$.
Conversely, the competing methods are unable to achieve the nominal 95\% coverage due to elevated empirical standard errors of the corresponding estimators, illustrating that naive estimation methods combined with machine learning does not guarantee valid statistical inference.
Exhibiting minimal sensitivity to covariate dimensionality and model misspecification, the proposed robust estimators consistently outperform both moment- and regression-based methods across all simulation settings.

\begin{figure}[H]
\centering

\begin{subfigure}[b]{0.400\textwidth}
\centering
\includegraphics[width=\textwidth, height=3.5cm ]{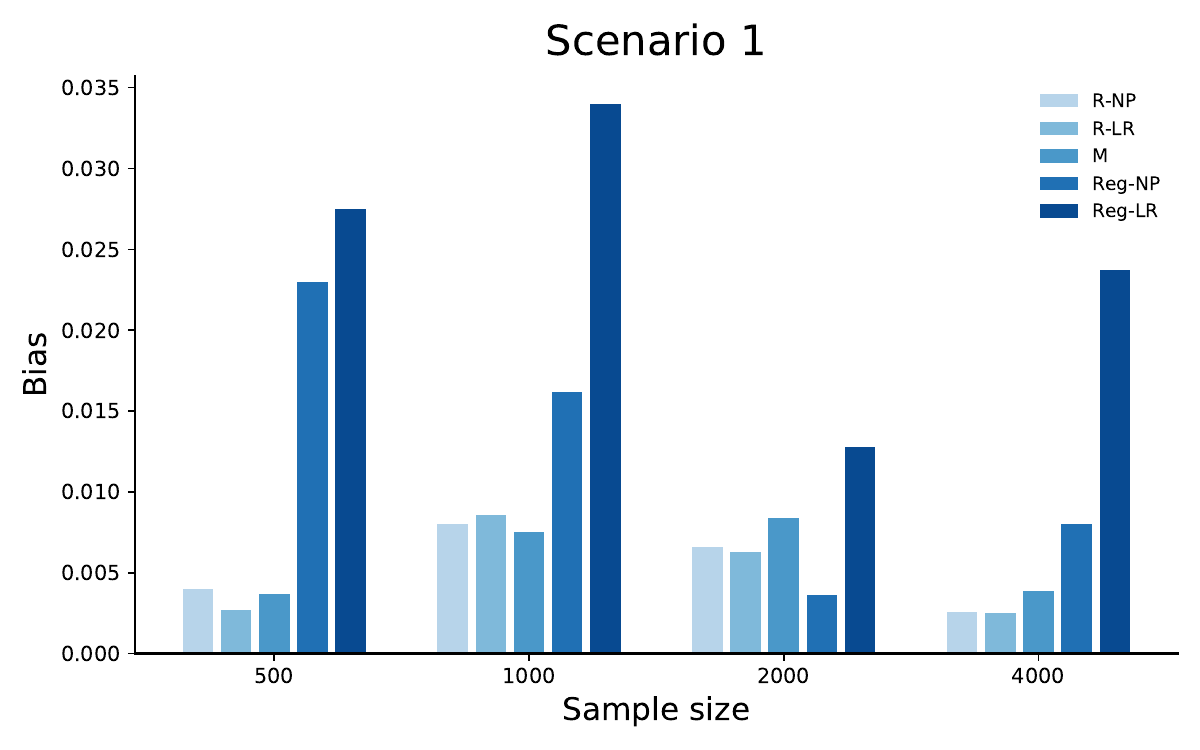}  
\end{subfigure} 
\begin{subfigure}[b]{0.400\textwidth}
\centering
\includegraphics[width=\textwidth, height=3.5cm ]{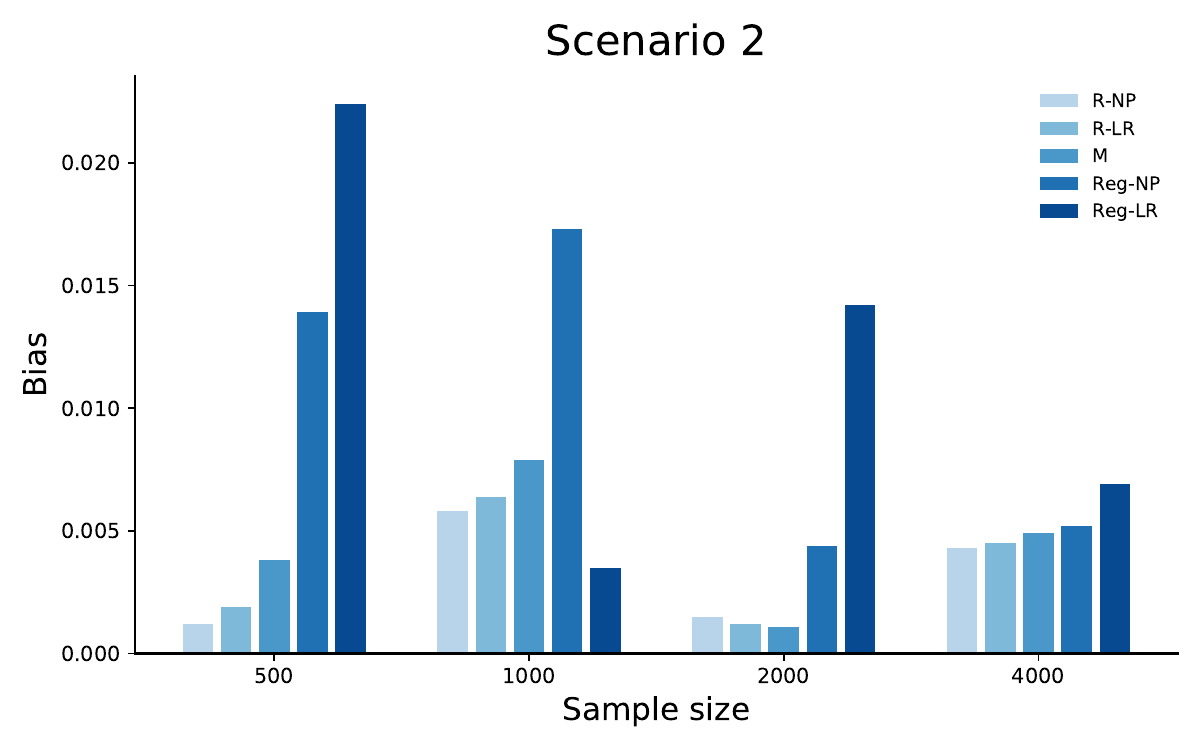} 
\end{subfigure} 

\begin{subfigure}[b]{0.400\textwidth}
\centering
\includegraphics[width=\textwidth, height=3.5cm ]{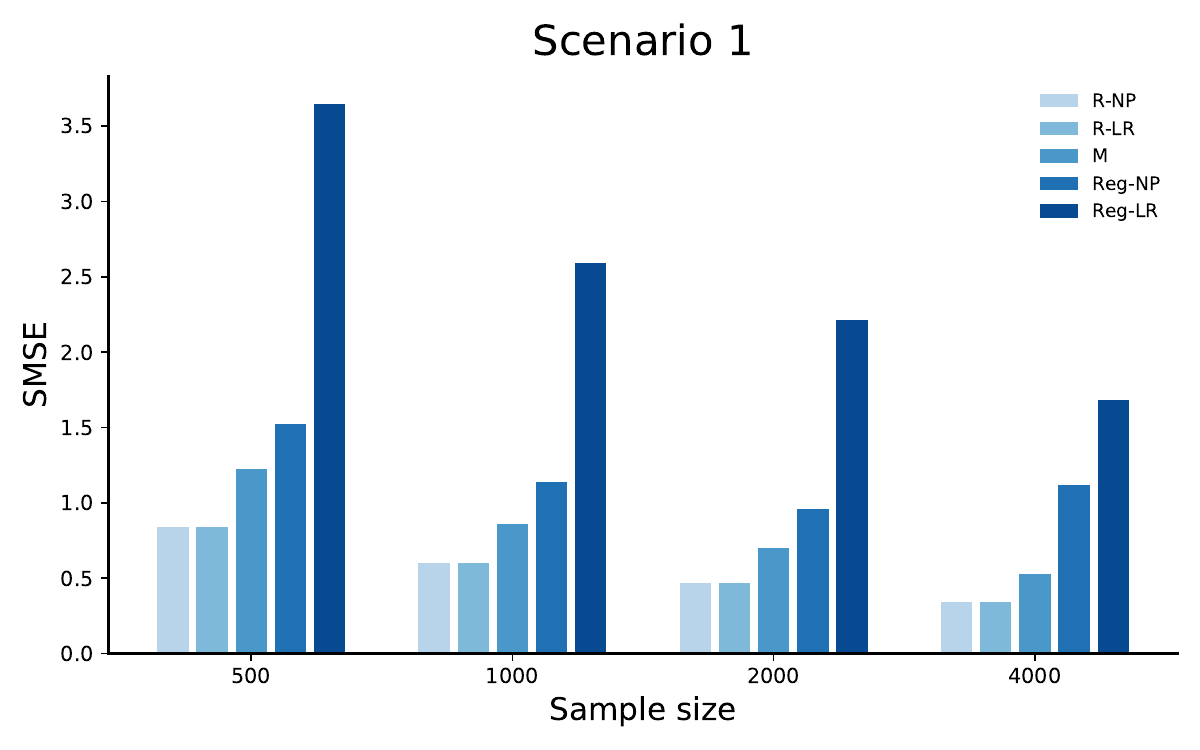} 
\end{subfigure} 
\begin{subfigure}[b]{0.400\textwidth}
\centering
\includegraphics[width=\textwidth, height=3.5cm ]{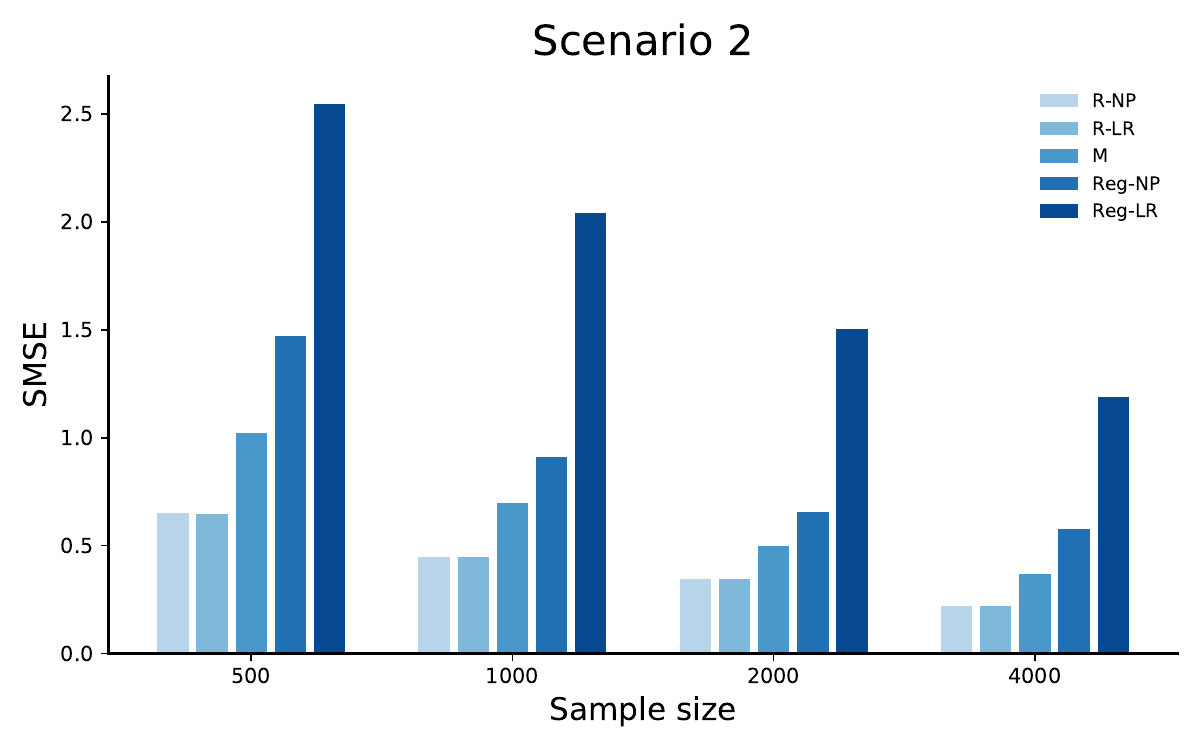}  
\end{subfigure}

\caption{Comparison of Bias and SMSE for five estimation methods with a covariate dimension of $p=4$ across sample sizes $n \in \{500, 1000, 2000, 4000\}$. 
The results are averaged over 500 replications. 
In each subfigure, the bars from left to right represent the metrics for the five methods: our proposed robust methods (R-NP and R-LR), the moment method (M), and the regression imputation methods (Reg-NP and Reg-LR).}
\label{fig2}
\end{figure}

\begin{figure}[H]
\centering

\begin{subfigure}[b]{0.400\textwidth}
\centering
\includegraphics[width=\textwidth, height=3.5cm ]{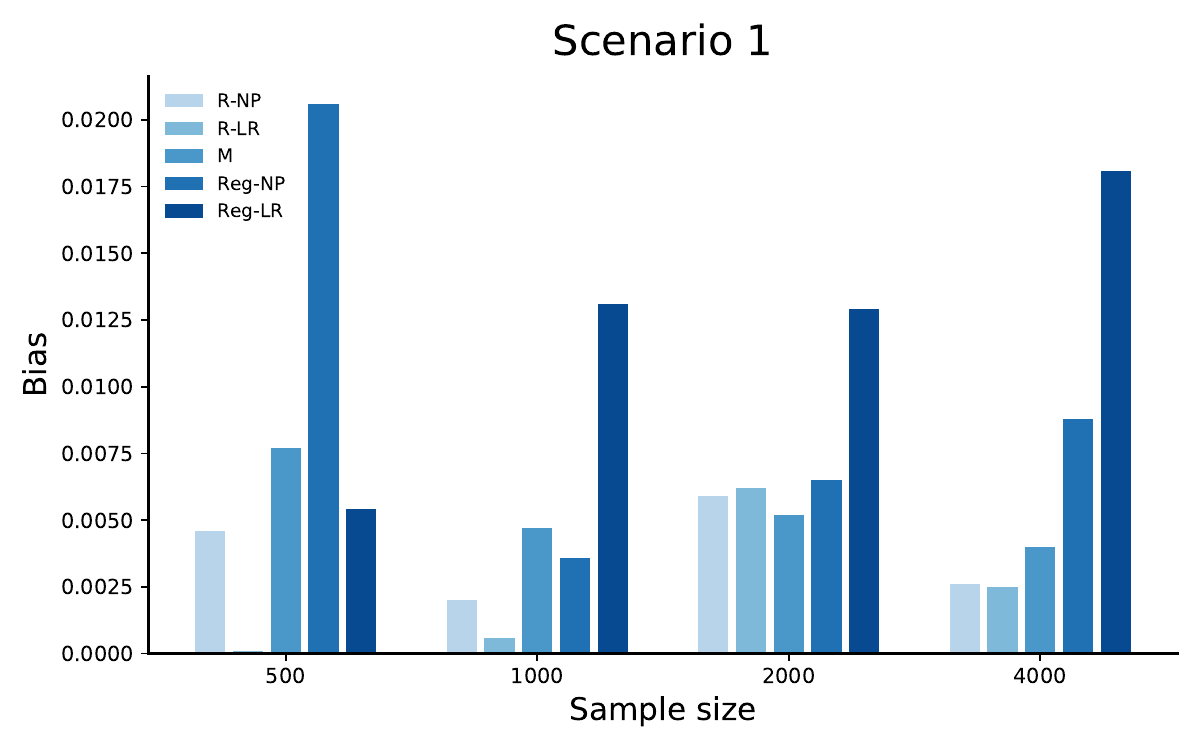}  
\end{subfigure} 
\begin{subfigure}[b]{0.400\textwidth}
\centering
\includegraphics[width=\textwidth, height=3.5cm ]{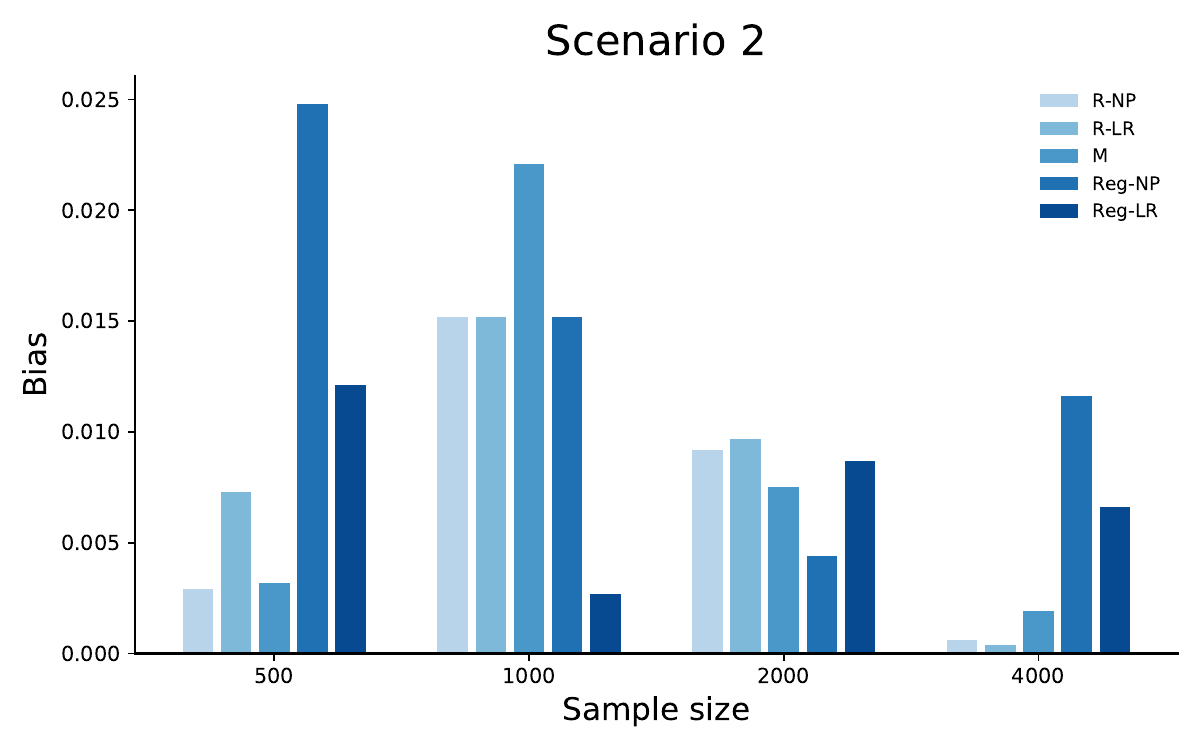}  
\end{subfigure} 

\begin{subfigure}[b]{0.4\textwidth}
\centering
\includegraphics[width=\textwidth, height=3.5cm ]{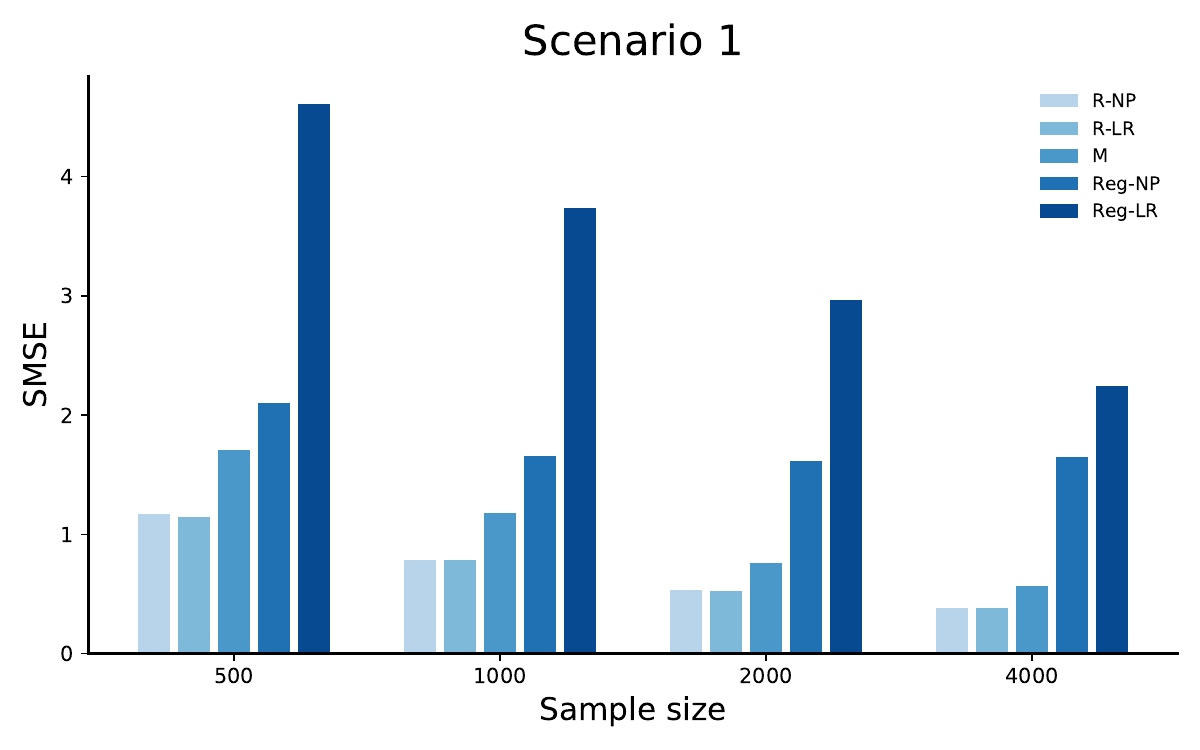} 
\end{subfigure} 
\begin{subfigure}[b]{0.400\textwidth}
\centering
\includegraphics[width=\textwidth, height=3.5cm ]{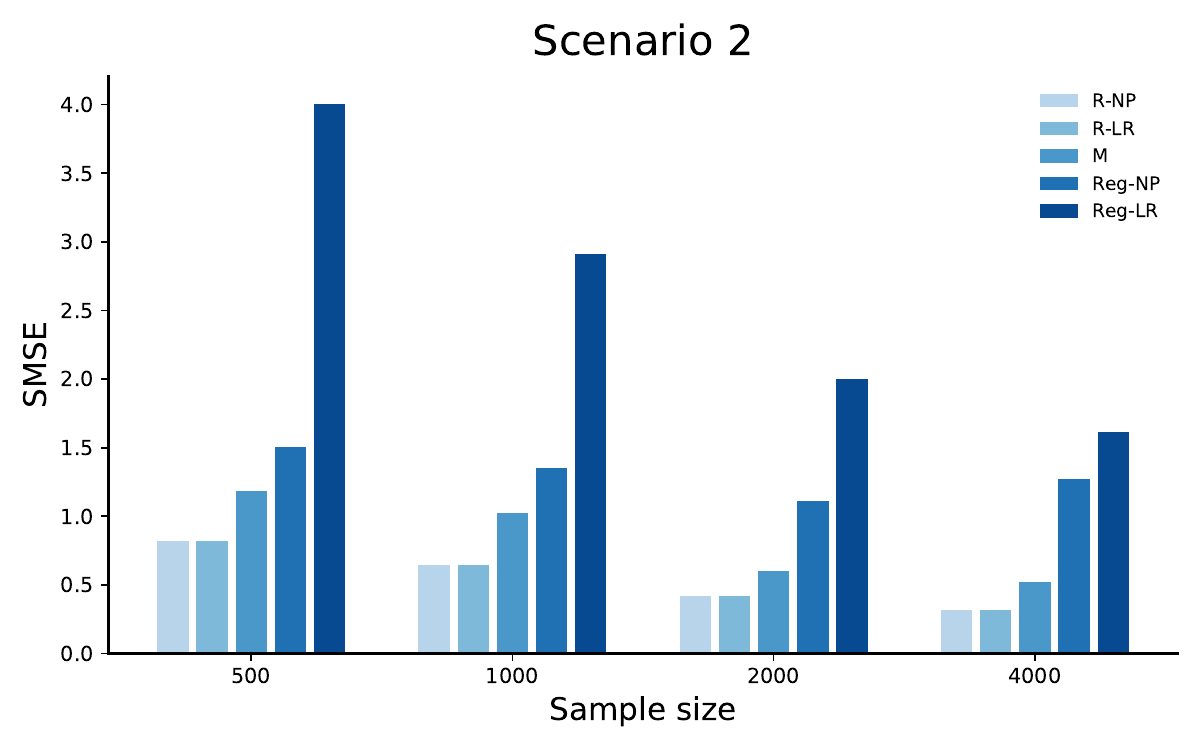}  
\end{subfigure}

\caption{Comparison of Bias and SMSE for five estimation methods with a covariate dimension of $p=10$ across sample sizes $n \in \{500, 1000, 2000, 4000\}$. 
The results are averaged over 500 replications. 
In each subfigure, the bars from left to right represent the metrics for the five methods: our proposed robust methods (R-NP and R-LR), the moment method (M), and the regression imputation methods (Reg-NP and Reg-LR).}
\label{fig3}
\end{figure}

\begin{table}[H]
\centering 
\caption{Comparison of coverage probabilities for five estimation methods with a covariate dimension of $p=4$ across sample sizes $n \in \{500, 1000, 2000, 4000\}$. The results are based on 500 replications.}
\resizebox{0.9\textwidth}{!}{
\begin{tabular}{@{}c ccccc | ccccc@{}} 
\hline
Sample Size & \multicolumn{5}{c}{Scenario 1} & \multicolumn{5}{c}{Scenario 2} \\
\cline{2-6} \cline{7-11}
& R-NP & R-LR & M & Reg-NP & Reg-LR & R-NP & R-LR & M & Reg-NP & Reg-LR \\
\hline
500  & 0.968 & 0.970 & 0.940 & 0.910 & 0.742 & 0.968 & 0.970 & 0.912 & 0.856 & 0.726 \\
1000 & 0.970 & 0.968 & 0.938 & 0.922 & 0.819 & 0.968 & 0.968 & 0.914 & 0.864 & 0.706 \\
2000 & 0.964 & 0.964 & 0.928 & 0.890 & 0.692 & 0.966 & 0.968 & 0.912 & 0.880 & 0.690 \\
4000 & 0.966 & 0.966 & 0.910 & 0.822 & 0.652 & 0.962 & 0.962 & 0.915 & 0.866 & 0.638 \\
\hline
\end{tabular}
}
\label{lab1}
\end{table}

\begin{table}[H]
\centering 
\caption{Comparison of coverage probabilities for five estimation methods with a covariate dimension of $p=10$ across sample sizes $n \in \{500, 1000, 2000, 4000\}$. The results are based on 500 replications.}
\resizebox{0.9\textwidth}{!}{
\begin{tabular}{@{}cccccc|cccccc@{}}
\hline
Sample Size & \multicolumn{5}{c}{Scenario 1} & \multicolumn{5}{c}{Scenario 2} \\
\cline{2-6} \cline{7-11}
& R-NP & R-LR & M & Reg-NP & Reg-LR & R-NP & R-LR & M & Reg-NP & Reg-LR \\
\hline
500  & 0.952 & 0.954 & 0.908 & 0.882 & 0.736 & 0.962 & 0.962 & 0.918 & 0.882 & 0.636 \\
1000 & 0.966 & 0.966 & 0.922 & 0.862 & 0.692 & 0.944 & 0.944 & 0.890 & 0.848 & 0.658 \\
2000 & 0.966 & 0.966 & 0.944 & 0.842 & 0.642 & 0.958 & 0.960 & 0.914 & 0.816 & 0.674 \\
4000 & 0.958 & 0.958 & 0.920 & 0.826 & 0.644 & 0.956 & 0.956 & 0.874 & 0.776 & 0.606 \\
\hline
\end{tabular}
}
\label{lab2}
\end{table}

\section{Case study: The Oregon Health Insurance Experiment}\label{Sec5}
The Oregon Health Insurance Experiment (OHIE) was established to evaluate the causal effects of expanding public health insurance (Medicaid) on healthcare utilization, health outcomes, financial strain, and overall well-being among low-income adults. 
Due to budget constraints, the state of Oregon could not accommodate all eligible individuals. 
To fairly allocate the limited number of available enrollment slots, officials established a waitlist and employed a randomized lottery system among the waitlisted applicants.

Let $Z \in \{0, 1\}$ denote the instrumental variable for randomized lottery selection, where $Z=1$ if a participant was selected to apply for Medicaid, and $Z=0$ otherwise. 
The selected participants could either apply for Medicaid coverage or forgo the opportunity.
Let $D \in \{0, 1\}$ denote the treatment indicator for Medicaid coverage, where $D = 1$ if the selected participant enrolled in the Medicaid program, and $D=0$ otherwise.  

After the experiment, researchers conducted in-person interviews with a subset of the OHIE participants to collect data on health insurance coverage, healthcare needs and experiences, physical measurements, and out-of-pocket medical costs (data accessible via https://www.nber.org/oregon/). 
The interview sample consists of 5,842 adults who were not selected ($Z = 0$) and 6,387 adults who were selected ($Z = 1$). 
Within the selected group, 1,957 individuals had successfully applied and enrolled in Medicaid ($D = 1$). 
Furthermore, the dataset includes a $p=32$-dimensional vector of controlled covariates ($X$), categorized into:
(i) demographic characteristics: financial burden, age, race, and education.
(ii) medical history.  
The randomized lottery ensures that the instrumental variable $Z$ is independent of the covariates $X$. 
After excluding records with missing responses, the final analytical sample comprises 12,141 individuals. 

To evaluate the impact of Medicaid, we focus on the causal effect of Medicaid coverage ($D$) on each of four distinct response variables ($Y$), including healthcare expenditures, clinical health indices, and subjective well-being: 
\begin{itemize}
\item Total medical spending (\textit{TMS}): A continuous measure representing aggregate medical expenditures, denominated in thousands of dollars. 
This variable captures the financial intensity of healthcare utilization.
\item Mental component score (\textit{MCS}): 
A continuous health-related quality-of-life metric. 
To ensure numerical stability, the score is scaled by a factor of $1/10$.
\item Physical component score (\textit{PCS}): 
A continuous metric assessing physical functional status. 
Similar to the \textit{MCS}, this variable is scaled by $1/10$.
\item Happiness (\textit{Hap}): 
A discrete, self-reported metric utilized as a proxy for the participants' overall subjective well-being.
\end{itemize}
Recall that we are interested in the causal effect $\beta_0= E \{(Y(1)-Y(0))(D(1)-D(0))\}$.

Previous studies utilized these data to evaluate the LATE of Medicaid on various socioeconomic and health outcomes.  
To investigate potential heterogeneous effects, we conducted subgroup analyses stratified by age into three distinct cohorts: 35 and under, 36 to 49, and 50 and older.
While \cite{qiu2021inference} also adopted a similar age-stratified analysis, our proposed approach differs from theirs.
In particular, \cite{qiu2021inference} employed a high-dimensional linear model, developing a debiased Lasso estimator based on the instrumental propensity score weighting method to infer the coefficients of Medicaid coverage. 
Their analysis revealed that the senior population with Medicaid coverage had lower out-of-pocket medical expenditures and higher levels of self-reported happiness. 
In contrast, we applied our proposed robust nonparametric method to separately examine the effects of Medicaid coverage on the four aforementioned outcomes across each age subgroup, thereby extending the evaluation beyond linear assumptions. 
Furthermore, to explore potential gender-specific heterogeneity, we conducted additional focused analyses on female and male participants, respectively. 

We applied the five estimation methods described in Section \ref{simulation} to quantify the causal effects of Medicaid coverage on the four outcomes by computing point estimates and 95\% confidence intervals for $\beta_0$. 
The estimation procedures and DNN configurations were identical to those detailed in Section \ref{simulation}. 
The results of our proposed R-NP approach are displayed in Figure \ref{fig1}, and the detailed results of five methods are summarized in Tables \ref{tab:pooled_results_by_age}, \ref{tab:female_results_by_age} and \ref{tab:male_results_by_age}.
The reported causal effects and their corresponding significance levels apply specifically to the subpopulation of compliers. 
Several important insights are listed below:

\begin{itemize}
\item[(i)] 
\textbf{Total medical spending (\textit{TMS}):}
Consistent with prior research (\cite{qiu2021inference}), our results confirm that Medicaid coverage led to a significant reduction in financial strain for senior groups (aged 50 and older). 
Moreover, we reveal evidence of decreased medication expenditures among younger adults (aged 35 and under) and female participants, reflecting substantial financial relief for these beneficiaries.

\item[(ii)] 
\textbf{Physical component score (\textit{PCS}):} 
The causal effect of Medicaid coverage on physical health outcomes was not significant across any of the evaluated age groups. 
This aligns with and broadens the findings of \cite{qiu2021inference}, who reported similar non-significant results specifically for older women.

\item[(iii)]  
\textbf{Mental component score (\textit{MCS}):}
We observe that Medicaid coverage yields significant improvements in mental health among middle-aged women (36–49 years).
This benefit was not identified in previous analyses (\cite{baicker2013oregon, qiu2021inference}).

\item[(iv)]  
\textbf{Happiness (\textit{Hap}):} 
The effect of Medicaid coverage on self-reported happiness was significant for both middle-aged women and the pooled female cohort. 
This contrasts with \cite{qiu2021inference}, who found no significant effects across age and gender subgroups. 
This discrepancy is possibly attributable to the robustness of our proposed method to the complexities of nuisance function estimation.

\item[(v)]
\textbf{Male subgroups:}
No significant effects can be found in male subgroups, as shown in Panel (c) of Figure~\ref{fig1}.

\item[(vi)] 
\textbf{Methodological comparisons:}
The point estimates and 95\% confidence intervals for four aforementioned outcomes of five methods are displayed in Tables \ref{tab:pooled_results_by_age}, \ref{tab:female_results_by_age} and \ref{tab:male_results_by_age}, which focus on the pooled, female and male groups, separately. 
These results underscore the necessity of robust estimation approaches. 

The regression-imputation-based methods (Reg-NP, Reg-LR) exhibit severe instability, yielding extreme point estimates. 
This instability shows their sensitivity to the prediction errors arising from nuisance parameter estimations.
Furthermore, while the robust linear approach (R-LR) mitigates this sensitivity, it suffers from the linear assumptions and the risk of model misspecification.

On the contrary, by integrating the proposed Neyman-orthogonal score function with deep learning methods, the proposed R-NP method captures complex, non-linear underlying data structures.
Therefore, the proposed estimator exhibits robustness and provides theoretically valid evaluations of the true causal effects.

\end{itemize}
Overall, our findings advance the existing literature by offering both a robust methodological tool and novel empirical evidence.

\begin{figure}[H]
\centering   
\begin{subfigure}[b]{1\textwidth}
\centering
\caption{Pooled population} 
\includegraphics[width=8cm,height=3.45cm]{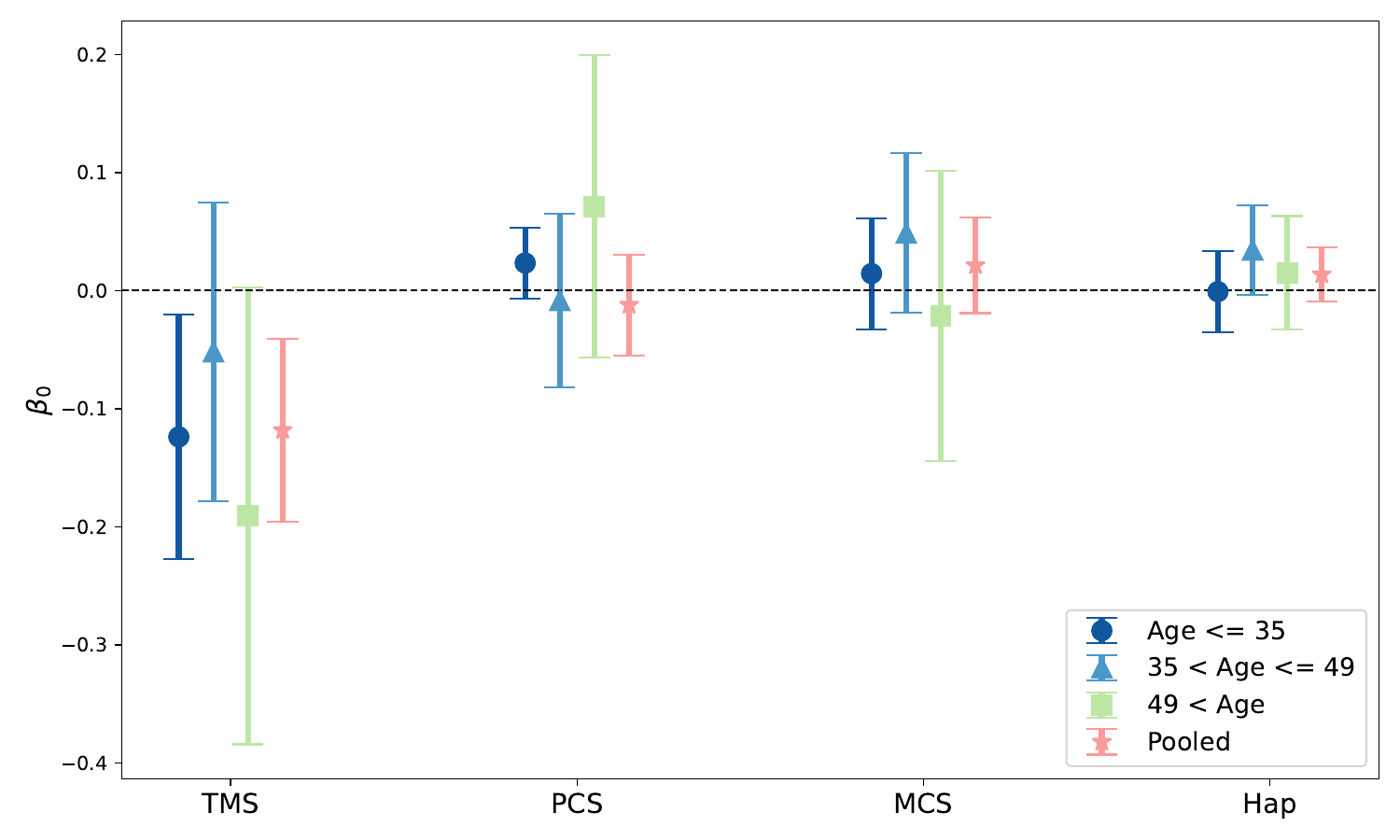}

\end{subfigure}

\vspace{1em} 

\begin{subfigure}[b]{1\textwidth}
\centering
\caption{Female population}  
\includegraphics[width=8cm,height=3.45cm]{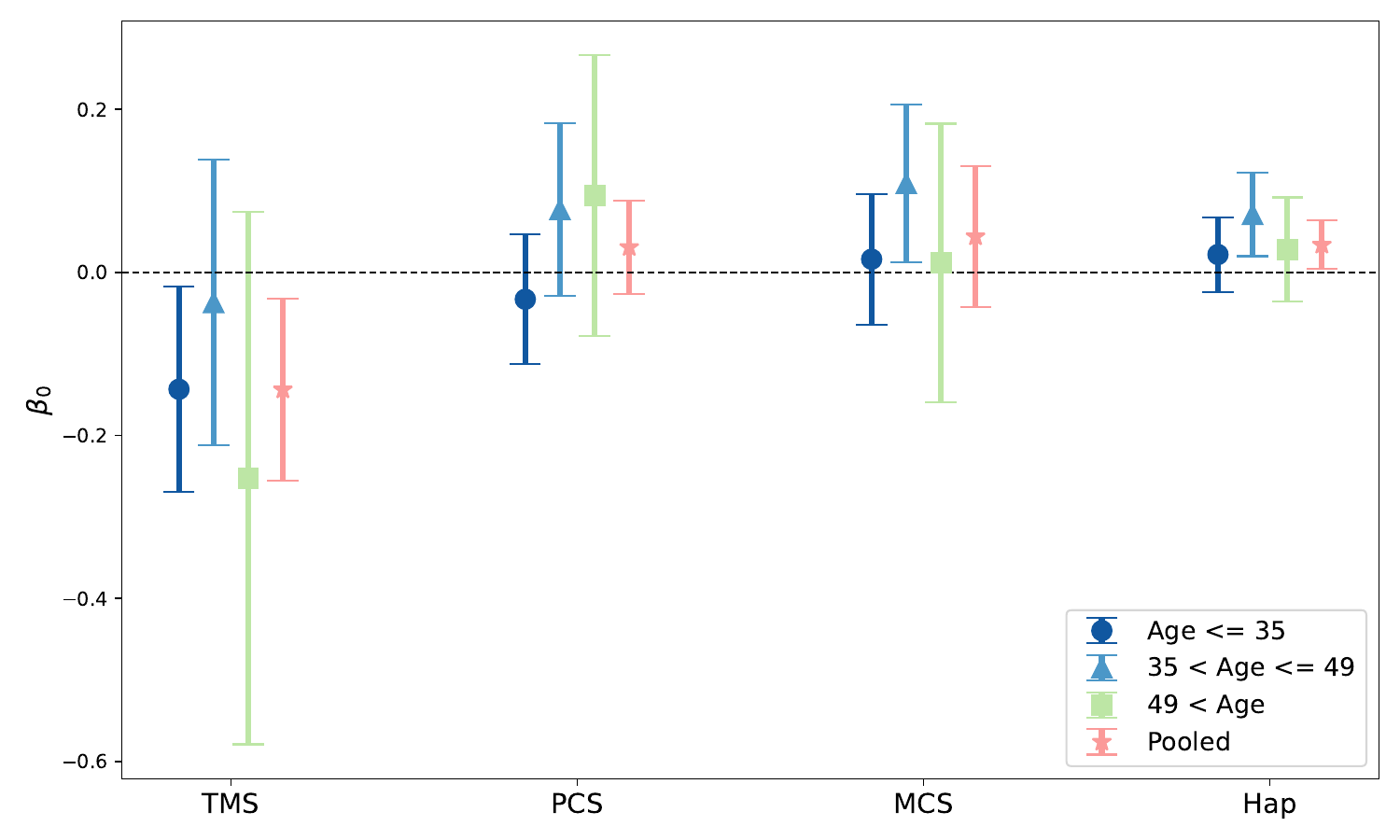}

\end{subfigure}

\vspace{1em}

\begin{subfigure}[b]{1\textwidth}
\centering
\caption{Male population} 
\includegraphics[width=8cm,height=3.45cm]{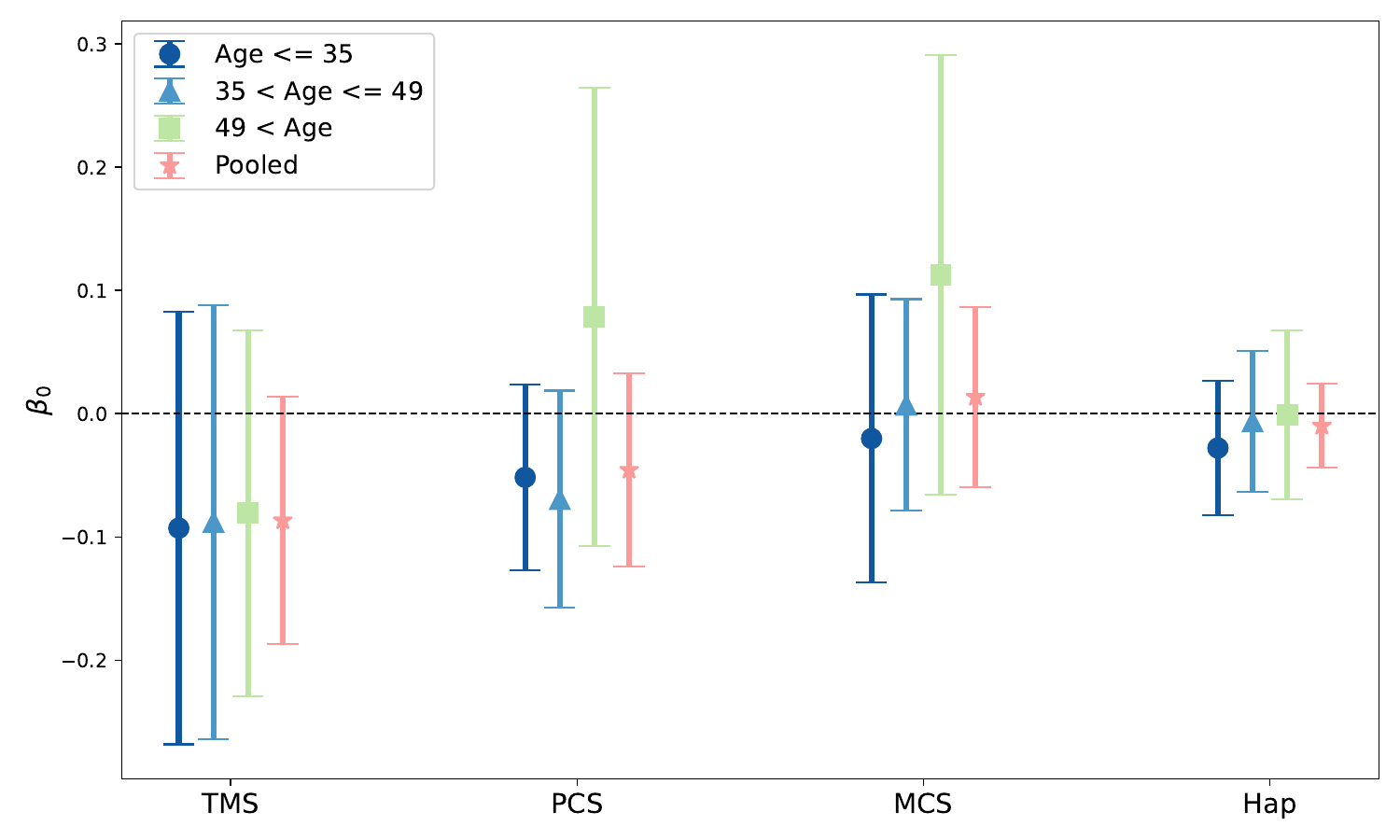}

\end{subfigure}

\caption{Proposed R-NP estimates and 95\% confidence intervals for the causal effects of Medicaid coverage on four outcomes in the OHIE dataset. 
Subfigure (a) presents results for the pooled population, while (b) and (c) focus on females and males, respectively. 
In each panel, colors represent subgroup estimates by age: indigo for individuals aged $\le$35, blue for ages 36–49, green for ages >49, and pink for the total group within that population.}
\label{fig1}
\end{figure}

\begin{table}[h]
\centering
\caption{Point estimates ($\hat{\beta}$) and 95\% confidence intervals (CI) for five estimation methods and the pooled population across different age groups.}
\resizebox{0.9\textwidth}{!}{ 
\begin{tabular}{@{}ll cccccccccc@{}}
\hline
& & \multicolumn{2}{c}{R-NP} & \multicolumn{2}{c}{R-LR} & \multicolumn{2}{c}{M} & \multicolumn{2}{c}{Reg-NP} & \multicolumn{2}{c}{Reg-LR} \\
\cline{3-4} \cline{5-6} \cline{7-8} \cline{9-10} \cline{11-12}
Age & Target & $\hat{\beta}$ & CI & $\hat{\beta}$ & CI & $\hat{\beta}$ & CI & $\hat{\beta}$ & CI & $\hat{\beta}$ & CI \\
\hline
All Ages & TMS & -0.118 & [-0.194, -0.042] & -0.065 & [-0.138, 0.008] & -0.105 & [-0.181, -0.029] & -0.099 & [-0.175, -0.023] & -0.799 & [-0.872, -0.726] \\
& PCS & -0.012 & [-0.053, 0.029] & 0.020 & [0.000, 0.040] & 0.042 & [0.001, 0.083] & -0.536 & [-0.577, -0.495] & -0.091 & [-0.111, -0.071] \\
& MCS & 0.021 & [-0.018, 0.060] & 0.025 & [0.000, 0.050] & 0.080 & [0.041, 0.119] & -0.169 & [-0.208, -0.130] & -0.122 & [-0.147, -0.097] \\
& Hap & 0.013 & [-0.009, 0.035] & -0.004 & [-0.024, 0.016] & 0.022 & [0.000, 0.044] & 0.047 & [0.025, 0.069] & -0.176 & [-0.196, -0.156] \\
\hline
$\text{Age}\le 35$ & TMS & -0.123 & [-0.225, -0.021] & -0.063 & [-0.159, 0.033] & -0.123 & [-0.225, -0.021] & -0.100 & [-0.202, 0.002] & -0.170 & [-0.266, -0.074] \\
& PCS & 0.023 & [-0.006, 0.052] & 0.016 & [-0.002, 0.034] & 0.007 & [-0.022, 0.036] & -7.564 & [-7.593, -7.535] & -0.012 & [-0.030, 0.006] \\
& MCS & 0.014 & [-0.033, 0.061] & 0.022 & [-0.021, 0.065] & 0.058 & [0.011, 0.105] & -1.990 & [-2.037, -1.943] & 0.049 & [0.006, 0.092] \\
& Hap & 0.000 & [-0.033, 0.033] & -0.008 & [-0.039, 0.023] & 0.006 & [-0.027, 0.039] & -0.028 & [-0.061, 0.005] & -0.095 & [-0.126, -0.064] \\
\hline
$35<\text{Age}\le 49$  & TMS & -0.051 & [-0.176, 0.074] & -0.043 & [-0.161, 0.075] & -0.042 & [-0.167, 0.083] & 0.002 & [-0.123, 0.127] & -0.462 & [-0.580, -0.344] \\
& PCS & -0.008 & [-0.081, 0.065] & 0.004 & [-0.029, 0.037] & 0.055 & [-0.018, 0.128] & -1.124 & [-1.197, -1.051] & -0.060 & [-0.093, -0.027] \\
& MCS & 0.048 & [-0.019, 0.115] & 0.037 & [-0.006, 0.080] & 0.111 & [0.044, 0.178] & 0.045 & [-0.022, 0.112] & 0.073 & [0.030, 0.116] \\
& Hap & 0.034 & [-0.003, 0.071] & 0.005 & [-0.026, 0.036] & 0.049 & [0.012, 0.086] & -0.011 & [-0.048, 0.026] & -0.170 & [-0.201, -0.139] \\
\hline
$49<\text{Age}$ & TMS & -0.190 & [-0.382, 0.002] & -0.120 & [-0.306, 0.066] & -0.180 & [-0.372, 0.012] & -0.023 & [-0.215, 0.169] & -0.003 & [-0.189, 0.183] \\
& PCS & 0.071 & [-0.056, 0.198] & 0.069 & [0.026, 0.112] & 0.111 & [-0.016, 0.238] & 0.175 & [0.048, 0.302] & 0.067 & [0.024, 0.110] \\
& MCS & -0.021 & [-0.143, 0.101] & 0.033 & [-0.022, 0.088] & 0.079 & [-0.043, 0.201] & 0.102 & [-0.020, 0.224] & 0.196 & [0.141, 0.251] \\
& Hap & 0.015 & [-0.032, 0.062] & -0.016 & [-0.055, 0.023] & 0.021 & [-0.026, 0.068] & 0.036 & [-0.011, 0.083] & 0.134 & [0.095, 0.173] \\
\hline
\end{tabular}
}
\label{tab:pooled_results_by_age}
\end{table}

\begin{table}[h]
\centering
\caption{Point estimates ($\hat{\beta}$) and 95\% confidence intervals (CI) for five estimation methods and the female population across different age groups.}
\resizebox{0.9\textwidth}{!}{ 
\begin{tabular}{@{}ll cccccccccc@{}}
\hline
& & \multicolumn{2}{c}{R-NP} & \multicolumn{2}{c}{R-LR} & \multicolumn{2}{c}{M} & \multicolumn{2}{c}{Reg-NP} & \multicolumn{2}{c}{Reg-LR} \\
\cline{3-4} \cline{5-6} \cline{7-8} \cline{9-10} \cline{11-12}
Age Group & Target & $\hat{\beta}$ & CI & $\hat{\beta}$ & CI & $\hat{\beta}$ & CI & $\hat{\beta}$ & CI & $\hat{\beta}$ & CI \\
\hline
All Ages & TMS & -0.143 & [-0.253, -0.033] & -0.099 & [-0.205, 0.007] & -0.124 & [-0.234, -0.014] & -0.139 & [-0.249, -0.029] & 0.028 & [-0.078, 0.134] \\
& PCS & 0.030 & [-0.027, 0.087] & 0.027 & [0.002, 0.052] & 0.052 & [-0.005, 0.109] & -0.234 & [-0.291, -0.177] & 0.026 & [0.001, 0.051] \\
& MCS & 0.044 & [-0.040, 0.128] & 0.029 & [-0.008, 0.066] & 0.071 & [-0.013, 0.155] & -0.446 & [-0.530, -0.362] & 0.108 & [0.071, 0.145] \\
& Hap & 0.033 & [0.004, 0.062] & 0.022 & [-0.003, 0.047] & 0.038 & [0.009, 0.067] & 0.039 & [0.010, 0.068] & 0.003 & [-0.022, 0.028] \\
\hline
$\text{Age}\le 35$ & TMS & -0.143 & [-0.268, -0.018] & -0.128 & [-0.246, -0.010] & -0.147 & [-0.272, -0.022] & -0.130 & [-0.255, -0.005] & -0.428 & [-0.546, -0.310] \\
& PCS & -0.033 & [-0.111, 0.045] & 0.010 & [-0.029, 0.049] & 0.006 & [-0.072, 0.084] & 0.661 & [0.583, 0.739] & -0.043 & [-0.082, -0.004] \\
& MCS & 0.016 & [-0.062, 0.094] & 0.031 & [-0.026, 0.088] & 0.062 & [-0.016, 0.140] & -0.008 & [-0.086, 0.070] & 0.110 & [0.053, 0.167] \\
& Hap & 0.021 & [-0.024, 0.066] & 0.013 & [-0.028, 0.054] & 0.021 & [-0.024, 0.066] & -0.013 & [-0.058, 0.032] & 0.156 & [0.115, 0.197] \\
\hline
$35<\text{Age}\le 49$  & TMS & -0.037 & [-0.211, 0.137] & -0.025 & [-0.190, 0.140] & -0.024 & [-0.198, 0.150] & 0.109 & [-0.065, 0.283] & 0.295 & [0.130, 0.460] \\
& PCS & 0.077 & [-0.027, 0.181] & 0.024 & [-0.021, 0.069] & 0.099 & [-0.005, 0.203] & 0.246 & [0.142, 0.350] & 0.040 & [-0.005, 0.085] \\
& MCS & 0.109 & [0.013, 0.205] & 0.074 & [0.013, 0.135] & 0.131 & [0.035, 0.227] & 0.229 & [0.133, 0.325] & 0.206 & [0.145, 0.267] \\
& Hap & 0.071 & [0.020, 0.122] & 0.045 & [0.002, 0.088] & 0.077 & [0.026, 0.128] & 0.076 & [0.025, 0.127] & 0.045 & [0.002, 0.088] \\
\hline
$49<\text{Age}$ & TMS & -0.252 & [-0.577, 0.073] & -0.188 & [-0.496, 0.120] & -0.263 & [-0.588, 0.062] & -0.140 & [-0.465, 0.185] & -0.987 & [-1.295, -0.679] \\
& PCS & 0.093 & [-0.077, 0.263] & 0.060 & [0.003, 0.117] & 0.105 & [-0.065, 0.275] & 0.251 & [0.081, 0.421] & -0.001 & [-0.058, 0.056] \\
& MCS & 0.011 & [-0.160, 0.182] & -0.021 & [-0.101, 0.059] & 0.015 & [-0.156, 0.186] & 0.143 & [-0.028, 0.314] & 0.081 & [0.001, 0.161] \\
& Hap & 0.028 & [-0.035, 0.091] & 0.008 & [-0.045, 0.061] & 0.028 & [-0.035, 0.091] & 0.063 & [0.000, 0.126] & 0.106 & [0.053, 0.159] \\
\hline
\end{tabular}
}
\label{tab:female_results_by_age}
\end{table}

\begin{table}[h]
\centering
\caption{Point estimates ($\hat{\beta}$) and 95\% confidence intervals (CI) for five estimation methods and the male population across different age groups.}
\resizebox{0.9\textwidth}{!}{ 
\begin{tabular}{@{}ll cccccccccc@{}}
\hline
& & \multicolumn{2}{c}{R-NP} & \multicolumn{2}{c}{R-LR} & \multicolumn{2}{c}{M} & \multicolumn{2}{c}{Reg-NP} & \multicolumn{2}{c}{Reg-LR} \\
\cline{3-4} \cline{5-6} \cline{7-8} \cline{9-10} \cline{11-12}
Age Group & Target & $\hat{\beta}$ & CI & $\hat{\beta}$ & CI & $\hat{\beta}$ & CI & $\hat{\beta}$ & CI & $\hat{\beta}$ & CI \\
\hline
All Ages & TMS & -0.086 & [-0.186, 0.014] & -0.017 & [-0.109, 0.075] & -0.074 & [-0.174, 0.026] & -0.034 & [-0.134, 0.066] & 0.057 & [-0.035, 0.149] \\
& PCS & -0.045 & [-0.123, 0.033] & 0.009 & [-0.020, 0.038] & 0.026 & [-0.052, 0.104] & -21.792 & [-21.870, -21.714] & 0.012 & [-0.017, 0.041] \\
& MCS & 0.013 & [-0.060, 0.086] & 0.024 & [-0.015, 0.063] & 0.088 & [0.015, 0.161] & -0.742 & [-0.815, -0.669] & 0.108 & [0.069, 0.147] \\
& Hap & -0.009 & [-0.042, 0.024] & -0.039 & [-0.068, -0.010] & 0.001 & [-0.032, 0.034] & 0.010 & [-0.023, 0.043] & -0.008 & [-0.037, 0.021] \\
\hline
$\text{Age}\le 35$ & TMS & -0.092 & [-0.266, 0.082] & -0.012 & [-0.181, 0.157] & -0.087 & [-0.261, 0.087] & -0.068 & [-0.242, 0.106] & -0.357 & [-0.526, -0.188] \\
& PCS & -0.051 & [-0.125, 0.023] & -0.001 & [-0.050, 0.048] & 0.044 & [-0.030, 0.118] & -1.471 & [-1.545, -1.397] & 0.635 & [0.586, 0.684] \\
& MCS & -0.020 & [-0.136, 0.096] & -0.005 & [-0.070, 0.060] & 0.003 & [-0.113, 0.119] & 0.080 & [-0.036, 0.196] & -0.930 & [-0.995, -0.865] \\
& Hap & -0.027 & [-0.080, 0.026] & -0.047 & [-0.096, 0.002] & -0.014 & [-0.067, 0.039] & 0.004 & [-0.049, 0.057] & -0.106 & [-0.155, -0.057] \\
\hline
$35<\text{Age}\le 49$  & TMS & -0.087 & [-0.261, 0.087] & -0.096 & [-0.272, 0.080] & -0.071 & [-0.245, 0.103] & 0.002 & [-0.172, 0.176] & 0.853 & [0.677, 1.029] \\
& PCS & -0.069 & [-0.155, 0.017] & -0.010 & [-0.061, 0.041] & 0.005 & [-0.081, 0.091] & 0.039 & [-0.047, 0.125] & 0.014 & [-0.037, 0.065] \\
& MCS & 0.007 & [-0.077, 0.091] & 0.023 & [-0.042, 0.088] & 0.091 & [0.007, 0.175] & 0.345 & [0.261, 0.429] & 0.024 & [-0.041, 0.089] \\
& Hap & -0.006 & [-0.063, 0.051] & -0.029 & [-0.078, 0.020] & 0.015 & [-0.042, 0.072] & 0.020 & [-0.037, 0.077] & 0.063 & [0.014, 0.112] \\
\hline
$49<\text{Age}$ & TMS & -0.080 & [-0.227, 0.067] & -0.027 & [-0.162, 0.108] & -0.069 & [-0.216, 0.078] & -0.046 & [-0.193, 0.101] & -0.312 & [-0.447, -0.177] \\
& PCS & 0.078 & [-0.108, 0.264] & 0.053 & [-0.014, 0.120] & 0.114 & [-0.072, 0.300] & 0.106 & [-0.080, 0.292] & 0.138 & [0.071, 0.205] \\
& MCS & 0.112 & [-0.064, 0.288] & 0.056 & [-0.026, 0.138] & 0.150 & [-0.026, 0.326] & 0.349 & [0.173, 0.525] & 0.798 & [0.716, 0.880] \\
& Hap & 0.000 & [-0.069, 0.069] & -0.054 & [-0.113, 0.005] & 0.011 & [-0.058, 0.080] & 0.005 & [-0.064, 0.074] & 0.085 & [0.026, 0.144] \\
\hline
\end{tabular}
}
\label{tab:male_results_by_age}
\end{table}

\section{Conclusion}\label{Sec6}
We have proposed a novel way to construct Neyman-orthogonal score functions for robust semiparametric estimation and inference with many infinite-dimensional nuisance parameters and developed a general approach for analyzing multi-stage nonparametric estimation procedures.

We apply the proposed approach to a specific case: estimating causal effects in a population using a binary instrumental variable and covariates.
Extending the Abadie's estimator, we construct a novel robust estimator for the parameter of interest via a Neyman-orthogonal score function.
This orthogonality desensitizes the target parameter to the estimation errors of the nuisance functions, and ensures that the target parameter estimator remains asymptotically normal, provided the nuisance function estimators achieve $o_p(n^{-1/4})$-consistency. 
Therefore, this theoretical guarantee allows us to safely employ deep neural networks for nonparametric estimation tasks, thereby mitigating the risks of model misspecification and data heterogeneity.  
Both the theoretical and experimental results show the significant potential of the proposed approach in addressing complex causal inference challenges.
\begin{appendix}

\section{Notation}
For any measurable function $f$, the $L^p(P)$-norm is given by $\|f\|_{p}^p:=\int |f(w)|^p dP(w)$ for $p\in[1,\infty)$,
with its empirical counterpart $L^p(\mathbb{P}_{n})$-norm defined as $\|f\|^p_{n,p}:=\mathbb{P}_{n}|f(W)|^p$ with a sequence $w_1^n =(w_1,\dots,w_n)$.
For a vector $v$, we denote $\|v\|_p$ and $\|v\|_\infty$ to be $p$-norm and $\infty$-norm, respectively.
Then, for given positive number $\epsilon$, let $\mathcal{N}(\epsilon,\mathcal{H},\|\cdot\|)$ and $\mathcal{N}_{[]}(\epsilon,\mathcal{H},\|\cdot\|)$  be the $\epsilon$-covering and $\epsilon$-bracketing number of a function class $\mathcal{H}$ with respect to the norm $\|\cdot\|$, respectively.

Denote $a_n \lesssim b_n$ as $a_n \le cb_n$ for some constant $c>0$.

For any multivariate function $\psi(\beta, f, h; w)$, we use $\partial_{f} \psi(\beta, f, h; w)$ to denote the first-order partial derivative of $\psi$ with respect to its second argument; second-order and mixed partial derivatives are defined analogously.

Let $l^{\infty}(\Theta)$ be the space of bounded functionals on $\Theta$ under the supremum norm.
For $h_1\in\Theta$ with support $\mathcal{X}$ satisfying $\sup\limits_{x\in\mathcal{X}}|h_1(x)|:=\|h_1\|_{\infty}\leq1$, the sample derivative map $Q_n$ and its population version $Q$ in direction $h_1$ at $(\mu,f)$ are
\begin{align*}
Q_{n_2}(\mu,f)[h_1]=&\lim_{\delta\rightarrow0}\mathbb{P}_{n_2}\frac{m_2(\mu+\delta h_1,f;W)-m_2(\mu,f;W)}{\delta}\\
:= & \partial_{\mu}\mathbb{P}_{n_2} m_2(\mu,f;W)[h_1]= \mathbb{P}_{n_2}\{ \partial_{\mu}m_2(\mu,f;W) \cdot h_1 \},
\\
Q(\mu,f)[h_1]:=&\  \partial_{\mu} \mathbb P m_2(\mu,f;W)[h_1] =   \mathbb P \partial_{\mu} \{m_2(\mu,f;W) \cdot h_1\}.
\end{align*}
Furthermore, for $h_2\in\Theta$ and $h_3\in \mathcal{F}$, we define
\begin{align*} 
\partial_{\mu} Q (\mu,f)[h_1][h_2]=&\lim_{\delta\rightarrow 0}\frac{Q(\mu+\delta h_2,f)[h_1]-Q(\mu,f)[h_1]}{\delta}\\
:=&\partial^2_{\mu \mu} \mathbb P m_2(\mu,f;W)[h_1][h_2] = \mathbb P \{\partial^2_{\mu \mu} m_2(\mu,f;W) \cdot h_1 \cdot h_2 \},\\
\partial_{f} Q(\mu,f)[h_1][h_3]=&\lim_{\delta\rightarrow 0}\frac{Q(\mu,f+\delta h_3)[h_1]-Q(\mu,f)[h_1]}{\delta} \\
:=&\partial^2_{\mu f} \mathbb P m_2(\mu,f;W)[h_1][h_3]= \mathbb P \{\partial^2_{\mu f}  m_2(\mu,f;W) \cdot h_1 \cdot h_3 \}.
\end{align*}

\section{Poofs of Theorem \ref{NO_score}}

\begin{proof}
(i). 
Recall that 
\begin{equation}
h_0(x) = -\frac{E[\partial^2_{\beta f} m(\beta_0, f_0;W)\mid X=x]}{E[\partial^2_{f f} m(\beta_0, f_0;W)\mid X=x]}. \nonumber
\end{equation}
For any $f\in \mathcal{F}$,
\begin{align*}
&\partial_f \mathbb{P}\psi^*(\beta_0, f_0,h_0; W)[f-f_0] \\
= &\partial_f \mathbb{P}[\partial_{\beta} m(\beta_0,f_0;W) +   \partial_{f} m(\beta_0,f_0;W)h_0(X)] [f-f_0]\\
=& \mathbb{P} [\partial^2_{\beta f} m(\beta_0,f_0;W) \{f(X)-f_0(X)\}] + \mathbb{P}[ \partial^2_{ff} m(\beta_0,f_0;W) h_0(X) \{f(X)-f_0(X)\}] \\
=& \mathbb{P} [\partial^2_{\beta f} m(\beta_0,f_0;W) \{f(X)-f_0(X)\} ]\\
&- \mathbb{P}\bigg[ \partial^2_{f f} m(\beta_0,f_0;W)  \frac{E[\partial^2_{\beta f} m(\beta_0, f_0;W)\mid X ]}{E[\partial^2_{f f} m(\beta_0, f_0;W)\mid X ]}\{f(X)-f_0(X)\}\bigg] \\ 
=& \mathbb{P} [\partial^2_{\beta f} m(\beta_0,f_0;W) \{f(X)-f_0(X)\} ]\\
&- \mathbb{P}\bigg[ E [\partial^2_{f f} m(\beta_0,f_0;W) \mid X] \frac{E[\partial^2_{\beta f} m(\beta_0, f_0;W)\mid X ]}{E[\partial^2_{f f} m(\beta_0, f_0;W)\mid X ]}\{ f(X)-f_0(X)\} \bigg] \\ 
=& 0,
\end{align*}
where the fourth equality follows from the law of total expectation.
Then, for any $ h\in \mathcal{F} $, we have
\begin{align*}
&\partial_{h} \mathbb{P} \psi^*(\beta_0,f_0,h_0; W)[h-h_0]
=  \mathbb{P}[\partial_{f} m(\beta_0,f_0;W) \{h(X)-h_0(X)\}]  = 0,
\end{align*}
which follows from the first-order condition of $m$ with respect to $f_0$.

(ii). 
For $\hat{\beta}_{n_2}$, by using Taylor expansion, we have
\begin{align}\label{Normal}
&-\sqrt{n_2}(\hat{\beta}_{n_2}  - \beta_0)\partial_{\beta} \mathbb{P}  \psi^*(\beta_0,f_0,h_0;W)  \nonumber \\
= &\sqrt{n_2} (\mathbb{P}_{n_2}-\mathbb{P})  \{\psi^*(\hat{\beta}_{n_2},\hat{f}_{n_1} ,\hat{h}_{n_1};W) - \psi^*(\beta_0,f_0,h_0;W) \} \nonumber\\
&  + \sqrt{n_2} \partial_{f} \mathbb{P} \psi^*(\beta_0,f_0,h_0;W)[\hat{f}_{n_1} -f_0] + \sqrt{n_2} \partial_{h} \mathbb{P} \psi^*(\beta_0,f_0,h_0;W)[\hat{h}_{n_1} -h_0] \nonumber \\
& + \sqrt{n_2} (\mathbb{P}_{n_2} - \mathbb{P}) \psi^*(\beta_0,f_0,h_0;W) + O_p(\|\hat{f}_{n_1} -f_0\|_2^2 + \|\hat{h}_{n_1} -h_0\|_2^2 + |\hat{\beta}_{n_2}  - \beta_0|^2),
\end{align}
where $o_p(n^{-1/4})$-consistent estimators $\hat{f}_{n_1}$ and $\hat{h}_{n_1}$ developed in the Step II of Algorithm \ref{alg:cross_fitting} approximate $f_0$ and $h_0$, respectively.
Let $\mathbb{P}_{n_2}$ denote the empirical measure based on the subsample $S_2$.
Consider the following decomposition,
\begin{align*}
&\sqrt{n_2} (\mathbb{P}_{n_2}-\mathbb{P}) \{\psi^*(\hat{\beta}_{n_2},\hat{f}_{n_1} ,\hat{h}_{n_1};W) - \psi^*(\beta_0,f_0,h_0;W) \} \\
=& \underbrace{\sqrt{n_2} (\mathbb{P}_{n_2}-\mathbb{P})  \{\psi^*(\beta_0,\hat{f}_{n_1} ,\hat{h}_{n_1};W) - \psi^*(\beta_0,f_0,h_0;W)\} }_{A}  \\
& + \underbrace{\sqrt{n_2} (\mathbb{P}_{n_2}-\mathbb{P})  \{\psi^*(\hat{\beta}_{n_2} ,\hat{f}_{n_1} ,\hat{h}_{n_1};W) - \psi^*(\beta_0,\hat{f}_{n_1} ,\hat{h}_{n_1};W)\} }_{B}.
\end{align*} 
Next, we show that $A=o_p(1)$ and $B=o_p(1)$.

First, due to sample splitting, the random sequence 

$\{\mathcal{E}_i :=\psi^*(\beta_0,\hat{f}_{n_1} ,\hat{h}_{n_1};W_i) - \psi^*(\beta_0,f_0,h_0;W_i)\}_{i=1}^{n_2}$ is i.i.d. given $S_1$.  
Hence, we have
\begin{align}\label{1.6}
& E \{ |A|^2 \mid S_1\} 
=E \left\{ \left[ \frac{1}{\sqrt{n_2}} \sum_{i=1}^{n_2}\{ \mathcal{E}_i - E[\mathcal{E}_i \mid S_1] \}\right]^2 \;\middle|\; S_1 \right\} \nonumber \\
=& \frac{1}{n_2}\sum_{i=1}^{n_2} E \left\{ (\mathcal{E}_i - E[ \mathcal{E}_i\mid S_1])^2 \mid S_1 \right\} 
= E  [|\mathcal{E}_i|^2 \mid S_1]  - |E [  \mathcal{E}_i \mid S_1] |^2 \nonumber \\
\leq & E \{ |\psi^*(\beta_0,\hat{f}_{n_1} ,\hat{h}_{n_1};W) - \psi^*(\beta_0,f_0,h_0;W)|^2 \mid S_1\},
\end{align}
where the second equality follows from the conditional independence between $\mathcal{E}_i - E[\mathcal{E}_i \mid S_1]$ and $\mathcal{E}_j - E[\mathcal{E}_j \mid S_1]$ given $S_1$ for $i \neq j$.
Applying Markov's inequality to the conditional distribution, we have
\begin{align*}
P ( | A  |>\epsilon  | S_1 ) \leq &	\frac{  E( |A|^2 \mid S_1) }{\epsilon^2} 
\lesssim \frac{  E ( |\hat{f}_{n_1}-f_0 |^2 + |\hat{h}_{n_1}-h_0 |^2 \mid S_1 ) }{\epsilon^2}.
\end{align*}
The second inequality follows from the smoothness of $\psi^*$ and the uniform boundedness of $\hat{\theta}_{n_1}$ and $\theta_0$, which imply the existence of a positive constant $C$ such that
\begin{align*}
E (|\mathcal{E}_i|^2 \mid S_1)  =&  E \{|\psi^*(\beta_0,\hat{f}_{n_1} ,\hat{h}_{n_1}  ;W_i) - \psi^*(\beta_0,f_0,h_0;W_i)|^2 \mid S_1 \} \\
\leq & C E ( |\hat{f}_{n_1}-f_0 |^2 + |\hat{h}_{n_1}-h_0 |^2 \mid S_1  ).
\end{align*}
Since $\hat{\theta}_{n_1} $ is a $o_p(n^{-1/4})$-consistent estimator of $\theta_0$, by the properties of conditional probability, we conclude that $A = o_p(1)$. 

Then, as $\hat{\beta}_{n_2}$ is $1$-dimensional and $o_p(n^{-1/4})$-consistent, applying Lemma 3.4.2 of \cite{vaart1996weak} yields that for $\delta_n = o(n^{-1/4})$,
\begin{align*}
E\bigg[\sup_{ |\beta - \beta_0| \leq \delta_n} |\sqrt{n_2}(\mathbb{P}_{n_2}-\mathbb{P}) [\psi^*(\beta, \hat{f}_{n_1} ,\hat{h}_{n_1};W) - \psi^*(\beta_0, \hat{f}_{n_1} ,\hat{h}_{n_1} ;W)]| \bigg| S_1 \bigg] \lesssim o_p(1).
\end{align*}
This implies that $B=o_p(1)$.

Therefore, combining \eqref{Normal} with the Neyman-orthogonal condition and $\sqrt{n_2} (\mathbb{P}_{n_2}-\mathbb{P}) [\psi^*(\hat{\beta}_{n_2},\hat{f}_{n_1} ,\hat{h}_{n_1};W) - \psi^*(\beta_0,f_0,h_0;W)]=A+B=o_p(1)$, we obtain that for $\hat{\beta}_{n_2}$, 
\begin{align*}
-\sqrt{n_2}(\hat{\beta}_{n_2}  - \beta_0)\partial_{\beta} \mathbb{P}\psi^*(\beta_0, f_0,h_0;W)  \nonumber
= \sqrt{n_2} (\mathbb{P}_{n_2} - \mathbb{P}) \psi^*(\beta_0,f_0,h_0;W) + o_p(1).
\end{align*}

Similar to the above arguments, exchanging roles of $S_1$ and $S_2$, we have
\begin{align*}
-\sqrt{n_1}(\hat{\beta}_{n_1}  - \beta_0)\partial_{\beta} \mathbb{P}\psi^*(\beta_0,f_0,h_0;W)  \nonumber
= \sqrt{n_1} (\mathbb{P}_{n_1} - \mathbb{P}) \psi^*(\beta_0, f_0,h_0;W) + o_p(1).
\end{align*}
Therefore, if $n_1 = n_2 = n/2$, then $\hat{\beta}_n^R = (\hat{\beta}_{n_1} + \hat{\beta}_{n_2}) /2$ satisfies
\begin{align*}
-\sqrt{n}(\hat{\beta}_{n}^R  - \beta_0)\partial_{\beta} \mathbb{P}\psi^*(\beta_0,f_0,h_0;W)  \nonumber
= \sqrt{n} (\mathbb{P}_{n} - \mathbb{P}) \psi^*(\beta_0,f_0,h_0;W) + o_p(1).
\end{align*} 
\end{proof} 	 

\section{\textbf{Proofs of Propositions \ref{prop0} and \ref{prop1}}}
\subsection{Proof of Proposition \ref{prop0}}
\begin{proof}
Under Assumption \ref{as_pre}(a) and (b), we have for any direction $h$,
\begin{align}\label{eq1}
&-\sqrt{n}[\mathbb{P}\{ \partial_{\beta} m(\hat{\beta}_n^M, \hat{f}_n^M; W) + \partial_{f} m(\hat{\beta}_n^M, \hat{f}_n^M; W)h \}\\ &-\mathbb{P} \{ \partial_{\beta} m(\beta_0, f_0; W) + \partial_{f} m(\beta_0, f_0; W) h  \}] \nonumber  \\
=&  \sqrt{n}(\mathbb{P}_n-\mathbb{P}) \left\{ \partial_{\beta} m(\beta_0, f_0; W) + \partial_{f} m(\beta_0, f_0; W)h \right\} + o_p(1).
\end{align}
Then, by Assumption \ref{as_pre}(c) and (d), we get that
\begin{align}\label{eq2}
&-\sqrt{n}[\mathbb{P}\{ \partial_{\beta} m(\hat{\beta}_n^M, \hat{f}_n^M; W) + \partial_{f} m(\hat{\beta}_n^M, \hat{f}_n^M; W)h \} \nonumber \\ 
&-\mathbb{P} \{ \partial_{\beta} m(\beta_0, f_0; W) + \partial_{f} m(\beta_0, f_0; W) h  \}] \nonumber \\
=&-\sqrt{n}\mathbb{P}[ \{ \partial^2_{\beta \beta} m(\beta_0, f_0; W) + \partial^2_{f \beta} m(\beta_0, f_0; W) h  \} (\hat{\beta}_n^M - \beta_0) ] \nonumber\\
& -\sqrt{n} \mathbb{P} \{ \partial^2_{\beta f} m(\beta_0, f_0; W) (\hat{f}^M_n-f_0) + \partial^2_{f f} m(\beta_0, f_0; W) h (\hat{f}^M_n-f_0) \}   + o_p(1).
\end{align}
Therefore, combining \eqref{eq1} and \eqref{eq2} completes the proof.
\end{proof}

\subsection{Proof of Proposition \ref{prop1}}
\begin{assumption}[Regularity conditions]\label{2S_norm}
Suppose that for any $h \in\Theta$, the following additional conditions hold:
\begin{itemize}
\item[(B1)] (Stochastic equicontinuity)
$$\sqrt{n_2}(Q_{n_2}-Q)(\hat{\mu}_{n_2}, \hat{f}_{n_1})[h] - \sqrt{n_2}(Q_{n_2}-Q)(\mu_0, f_0)[h] = o_p(1).$$
\item[(B2)] (Score condition)
\begin{enumerate}
\item[(i)] $Q(\mu_0, f_0)[h] = 0$.
\item[(ii)] $Q_{n_2}(\hat{\mu}_{n_2}, \hat{f}_{n_1})[h] = o_p(n^{-1/2})$.
\end{enumerate}
\item[(B3)] (Smoothness)The map $(\mu, f) \mapsto Q(\mu, f)[h]$ is Fréchet-differentiable at $(\mu_0, f_0)$.
\item[(B4)] (Second-order condition) 
\begin{align*}
&Q(\hat{\mu}_{n_2}, \hat{f}_{n_1})[h] - Q(\mu_0, f_0)[h] - \partial_{\mu} Q(\mu_0, f_0)[h][\hat{\mu}_{n_2}-\mu_0] \\
&- \partial_{f} Q(\mu_0, f_0)[h][\hat{f}_{n_1}-f_0] = o_p(n^{-1/2}).
\end{align*} 
\end{itemize}
\end{assumption}

\begin{proof}
By the decomposition of the score condition, we have
\begin{align}\label{eq0.0}
\sqrt{n_2} Q_{n_2}(\hat{\mu}_{n_2}, \hat{f}_{n_1})[h] = \sqrt{n_2} (Q_{n_2}-Q)(\hat{\mu}_{n_2}, \hat{f}_{n_1})[h] + \sqrt{n_2} Q(\hat{\mu}_{n_2}, \hat{f}_{n_1})[h].
\end{align}
Under the score condition (b)(ii), the left-hand side of \eqref{eq0.0} is $o_p(1)$.
For the first term on the right-hand side of \eqref{eq0.0}, the stochastic equicontinuity condition (a) implies that $$\sqrt{n_2} (Q_{n_2}-Q)(\hat{\mu}_{n_2}, \hat{f}_{n_1})[h] = \sqrt{n_2} (Q_{n_2}-Q)(\mu_0, f_0)[h] + o_p(1).$$
For the second term, using the score condition (b)(i), the smoothness condition (c), and the second-order condition (d), we get  \begin{align}\label{eq0.1} 
&\sqrt{n_2} Q(\hat{\mu}_{n_2}, \hat{f}_{n_1})[h] =\sqrt{n_2} \big( Q(\hat{\mu}_{n_2}, \hat{f}_{n_1})[h] - Q(\mu_0, f_0)[h] \big)\nonumber \\
=& \sqrt{n_2} \partial_{\mu} Q(\mu_0, f_0)[h][\hat{\mu}_{n_2}-\mu_0] + \sqrt{n_2} \partial_{f} Q(\mu_0, f_0)[h][\hat{f}_{n_1}-f_0] + o_p(1). 
\end{align}
Combining \eqref{eq0.0} and \eqref{eq0.1}, we have
\begin{align*} 
-\sqrt{n_2} \partial_{\mu} Q(\mu_0, f_0)[h][\hat{\mu}_{n_2}-\mu_0] =& \sqrt{n_2} \partial_{f} Q(\mu_0, f_0)[h][\hat{f}_{n_1}-f_0] \\ &+ \sqrt{n_2} (Q_{n_2}-Q)(\mu_0, f_0)[h] + o_p(1),
\end{align*}
which completes the proof.
\end{proof}
\end{appendix}

\bibliographystyle{unsrt}  
\bibliography{ref.bib}
\end{document}